\documentclass[10pt,twocolumn,letterpaper]{article}
\usepackage[T1]{fontenc}
\usepackage{mathptmx}                 
\usepackage[scaled=0.92]{helvet}
\usepackage{courier}
\usepackage[letterpaper,margin=0.75in,top=0.9in,bottom=1in,columnsep=0.25in]{geometry}
\usepackage[hyphens]{url}
\usepackage{graphicx}
\usepackage{natbib}
\usepackage{caption}
\usepackage{algorithm}
\usepackage{algorithmic}
\usepackage{booktabs}
\usepackage{array}
\usepackage{placeins}
\usepackage{xcolor}
\usepackage{amsmath,amssymb,amsthm}
\usepackage[colorlinks,allcolors=blue!55!black]{hyperref}
\graphicspath{{figures/}}

\newcolumntype{R}[1]{>{\raggedright\arraybackslash}p{#1}}
\providecommand{\texorpdfstring}[2]{#1}
\newcommand{\singlefigwidth}{0.92\columnwidth}
\newcommand{\panelfigwidth}{0.87\textwidth}
\newcommand{\widefigwidth}{0.97\textwidth}

\newif\ifqmi\qmifalse   
\newcommand{\xmark}{\ensuremath{\times}}   
\theoremstyle{plain}
\newtheorem{proposition}{Proposition}

\newcommand{\CUTscaleup}{At $n{=}12$ qubits, fusion matches the uncut reference within $0.03$
at \emph{every} $k\le6$ ($0.887{\pm}0.027$ vs.\ $0.860$--$0.880$) while $9^k$ climbs to
$5.3{\times}10^{5}$. At $n{=}16$ ($k{=}8$, $9^k{=}4.3{\times}10^{7}$) and $n{=}20$ ($k{=}10$,
$3.5{\times}10^{9}$), fusion still trains to $0.779{\pm}0.050$ and $0.829{\pm}0.058$
against $0.5$ chance ($5$ seeds).}
\newcommand{\CUTfaithful}{over $10$ seeds fusion is equivalent, at the predefined margin
$\delta{=}0.02$, to reconstruction and to a tuned classical model on the image tasks
(Fashion-MNIST $0.799$ vs.\ $0.799$/$0.797$; MNIST $0.990$ vs.\ $0.988$/$0.991$), and the
pattern persists at doubled feature dimension (PCA-8 on $4{+}4$-qubit registers: within
${\sim}0.01$ on both, equivalent on MNIST and inconclusive at $10$ seeds on Fashion-MNIST). On the hard synthetic task,
fusion ($0.91$) far exceeds the faithful Kawase sum ($0.56$, paired $p{=}0.002$), but here the
cross-cut residual is a small yet statistically detectable shortfall below reconstruction
($0.91$ vs.\ $0.93$; paired gap $-0.024$, $95\%$ CI $[-0.043,-0.003]$).}

\title{How Much Reconstruction Does Quantum Machine Learning Need?\\
       Late Fusion of Uncoupled Quantum Subcircuits}
\author{Prabhjot Singh, Adel N. Toosi, Rajkumar Buyya\\[0.6ex]
  {\normalsize Quantum Cloud Computing and Distributed Systems (qCLOUDS) Lab}\\
  {\normalsize School of Computing and Information Systems, The University of Melbourne, VIC 3010, Australia}\\
  {\normalsize\texttt{prabhjot.singh.1@student.unimelb.edu.au}}}
\date{}

\begin{document}
\makeatletter
\twocolumn[
\begin{@twocolumnfalse}
\maketitle
\begin{abstract}
Circuit cutting lets a large quantum neural network (QNN) run as independent subcircuits on
small devices, but rebuilding its outputs by \emph{quasiprobability reconstruction} carries
a classical sampling overhead \emph{exponential in the number of cuts}. This paper asks
whether machine-learning tasks need this step. We replace it with \emph{late fusion}:
subcircuits are trained without cross-cut coupling, and a small \emph{classical} head
combines their local measurements at a cost constant in the cut count. To characterize the
trade-off, we introduce a \emph{reconstruction dial}~$Q$, a physically executable budget
that retains a fraction of the quasiprobability terms, and a \emph{cut-entanglement
diagnostic} that indicates, within a fixed scale, how much reconstruction a task needs
(Spearman $\rho{=}0.59$ over $104$ runs). Across synthetic and standard datasets,
\emph{uncoupled late fusion is equivalent to full reconstruction within a predefined margin
of $0.02$ accuracy on most benchmarks and pooled over the controlled sweep, at exponentially
lower cost}; deficits of $0.02$--$0.04$ remain on the most cross-cut tasks. Fusion tops the
dial's accuracy/cost frontier, is more robust to shot and device noise, and loses nothing as
the coupled model's entangling depth grows. We state the exact condition under which fusion
must lose accuracy, a decision-relevant cross-cut correlator, and realize it with
entangled-state inputs; for classical inputs the register partition discards no information.
We claim no advantage over classical machine learning.
\end{abstract}
\vspace{1.2em}
\end{@twocolumnfalse}
]
\makeatother

\section{Introduction}

Variational quantum circuits, commonly called quantum neural networks (QNNs), are a leading
candidate application for near-term quantum hardware. Useful models, however, require more
qubits than today's devices reliably provide. \emph{Circuit
cutting}~\citep{peng2020,mitarai2021,piveteau2022} addresses this limitation by partitioning a
large circuit into subcircuits that each fit on a small device: the subcircuits are executed
separately, and their results are recombined classically. In turn, this approach has enabled
distributed QNN training~\citep{distributedestimator,marshall2023}.

\paragraph{The exponential cost of reconstruction.}
Recombination is not free. Cutting replaces each non-local gate by a signed
\emph{quasiprobability} mixture of local operations. Rebuilding the original circuit's
expectation values, which we call \emph{quasiprobability reconstruction}, then requires two
additional steps: (i)~running many extra subexperiments, and (ii)~combining them with a signed
sum whose statistical variance is inflated by a factor $\gamma^2$ per cut. Because $\gamma$ is
multiplicative across cuts, the \emph{sampling overhead grows exponentially in the number of
cuts}: it scales as $O(\gamma^{2k})$ for $k$ cuts and, for CNOT-type cuts, for instance, as
$O(9^k)$. Empirically, reconstruction dominated end-to-end runtime, accounting for
$53$--$95\%$ of it~\citep{distributedestimator}. It remains unclear, however, whether
machine-learning tasks need exact reconstruction at all, or whether a cheap, trainable
classical combination of the subcircuits' measurements can recover comparable accuracy at a
cost independent of the number of cuts.

\paragraph{Our proposal: late fusion.}
We import \emph{late} (decision-level) \emph{fusion} from multimodal deep
learning~\citep{baltrusaitis2019,snoek2005}. Put briefly, each subcircuit is treated as an
independent ``modality'': it is trained and measured on its own, and a small classical
\emph{fusion head} combines the resulting measurement features into a prediction. This design
replaces the fixed, exponential-cost quasiprobability sum with a learned classical function
that is optimized directly for the task loss, at a cost that is constant in the cut count and
linear in the number of subcircuits. Note that late fusion is \emph{not} an approximation of
reconstruction. Reconstruction restores the quantum correlation across the cut, whereas fusion
discards that correlation and relies on the head to compensate. Whether such compensation
suffices is a structural property of the task, which we make precise below.

\paragraph{A dial and a diagnostic.}
Rather than presenting fusion and reconstruction as an all-or-nothing choice, we expose a
\emph{reconstruction dial} $Q\in[0,1]$ that interpolates between pure fusion ($Q{=}0$) and
full reconstruction ($Q{=}1$) by retaining a tunable fraction of the reconstruction's
quasiprobability mass. Because the reconstruction overhead is lower-bounded by the
entanglement across the cut~\citep{jing2024}, we further introduce a \emph{cut-entanglement
diagnostic} that indicates how far along the dial a given task must be operated.

In this paper, we aim to shed light on whether a cut QNN needs quasiprobability
reconstruction for machine-learning tasks, or whether late fusion suffices. To this end, we
train subcircuits without cross-cut coupling and compare the fused readout head-to-head with
exact reconstruction of the same cut circuit. The comparison is made at matched cost, under
shot and device noise, and on live hardware. We then sweep the dial and evaluate the
diagnostic against it. Our
findings are threefold. First, late fusion matches reconstruction wherever the task's
information is local, while the overhead of reconstruction grows exponentially. Second,
fusion is more robust to shot and device noise. Third, fusion loses accuracy when the label is
carried by a decision-relevant cross-cut correlator, a regime that we realize explicitly with
entangled data. We chart these regimes explicitly. Throughout, tuned classical machine
learning matches the quantum model, so no quantum advantage is claimed.

\paragraph{Contributions.}
Independent quantum modules combined by a classical aggregator are an established
pattern~\citep{tuysuz2023,wu2022scalable,park2025multichip}, and that architecture is not our
claim. Our contribution is to place a \emph{learned} fusion readout inside a
\emph{circuit-cutting} pipeline as a drop-in alternative to quasiprobability reconstruction,
and to characterize when it suffices. We deliver the following main contributions:
\begin{itemize}
\item We introduce late fusion as a reconstruction alternative. We train subcircuits without
cross-cut coupling, so that no gradient flows through the cut, and combine their outputs with
a trainable classical head. We compare this \emph{uncoupled fusion} readout head-to-head with
exact reconstruction of the \emph{same} cut circuit at matched cost, under shot and device
noise, and on live hardware.
\item We introduce a reconstruction dial $Q$, a tunable reconstruction budget between the
coupled circuit's pure-fusion readout and its exact reconstruction, which reveals how much of
the exponential budget carries label information.
\item We introduce a cut-entanglement diagnostic, grounded in the bound that relates
reconstruction overhead to the entanglement cost of the cut. We evaluate it as a within-scale
indicator of how much reconstruction a task needs, and we report its requirements and its
negative results alongside.
\item We state a locality theorem and chart a regime map. We give an exact condition,
decision-relevant cross-cut correlation, under which fusion must lose accuracy, and we realize
it in a controlled entangled-data experiment. We also provide a map relating fusion,
reconstruction and classical machine learning, in which fusion's regime of local information
overlaps the regime where classical machine learning also suffices.
\end{itemize}
Prior cut-circuit readouts reconstruct exactly~\citep{peng2020}, truncate~\citep{marshall2023},
or sum without learning~\citep{kawase2024}. None of them compares a learned readout against
reconstruction at matched cost, exposes the continuum between the two, or states the condition
under which the learned readout must lose accuracy.

\paragraph{Scope.}
We study simulable scales because the exact reconstruction reference that anchors every
comparison exists only there, and we verify the overhead behavior through
\texttt{qiskit-addon-cutting}. We claim no quantum advantage over classical machine learning;
the value proposition is conditional. Given a circuit-cutting quantum machine learning (QML)
pipeline, we show which resource is wasted and provide the instrument that detects it. \ifqmi
\else The
remainder of the article reviews related work, presents the methodology and its supporting
theory, describes the experimental setup, reports the results, and discusses their implications
and limitations.\fi

\section{Related Work}

This section situates late fusion among four bodies of prior work: circuit cutting itself,
independent quantum modules with classical aggregation, readouts for cut circuits, and the
ensemble and multimodal-fusion literatures from which our combiner is borrowed. Lastly, we
describe the benchmarking practice that our evaluation adopts. Table~\ref{tab:related}
summarizes the readout strategies discussed here.

\paragraph{Circuit cutting and knitting.}
Circuit cutting decomposes a non-local operation into local pieces at an exponential sampling
price. More concretely, wire cutting~\citep{peng2020} and gate cutting~\citep{mitarai2021}
rewrite a non-local operation as a signed sum of local operations. The sampling overhead of
this decomposition is $\gamma^{2}$ per cut and multiplicative across cuts, which gives a total
overhead of $O(\gamma^{2k})$ for $k$ cuts~\citep{piveteau2022,brenner2023,ufrecht2024}.
Randomized measurements~\citep{lowe2023} lower the per-cut constant but remain exponential,
and classical shadows~\citep{huang2020shadows} give efficient local readouts but do not
perform cut reconstruction. \citet{jing2024} prove that this overhead is lower-bounded by the
entanglement cost across the cut; this result is the basis for our cut-entanglement
diagnostic. Systems realizations include CutQC~\citep{cutqc} and, for QNN training,
DistributedEstimator~\citep{distributedestimator}, which quantified reconstruction as the
dominant runtime cost. All of these methods \emph{reconstruct}. Our work aims to shed light on
what happens when, for machine-learning tasks, reconstruction is omitted.

\begin{table*}[t]
\caption{Readout strategies for cut-circuit QML and the closest uncut module-aggregation
methods (QPD: quasiprobability decomposition; reconstruction also
includes~\citep{mitarai2021,cutqc}). ``n/a'': the method does not cut a circuit, so it has no
reconstruction baseline to compare against. Cost is in the cut count $k$; the fusion
readout is instead linear in the number of subcircuits (two throughout). ``Indep.'': subcircuits trained with no gradient flow through
the cut. ``Dial'': tunable interpolation between fusion and full reconstruction. ``Diagnostic'':
a measurable indicator of how much reconstruction a task needs. $^{\dagger}$Requires the
trained coupled circuit; validated within a fixed scale only (Results).}
\label{tab:related}
\centering\footnotesize
\setlength{\tabcolsep}{3pt}
\begin{tabular}{@{}R{3.25cm}R{2.45cm}R{1.5cm}p{0.9cm}R{1.6cm}R{1.9cm}@{}}
\toprule
Method & Combiner & Cost in $k$ & Indep. & Dial & Diagnostic \\
\midrule
QPD reconstruction~\citep{peng2020} & exact signed QPD sum & $O(\gamma^{2k})$ & \xmark & \xmark & \xmark \\
DistributedEstimator \citep{distributedestimator} & exact QPD sum, distributed & $O(9^{k})$ & \xmark & \xmark & \xmark \\
\citet{marshall2023} & learned \emph{truncated} QPD sum & chosen budget & \xmark & \xmark\ (fixed $T$) & \xmark \\
\citet{marchisio2024} & mid-circuit re-encoding & linear & \xmark & \xmark & \xmark \\
Classical splitting~\citep{tuysuz2023} & fixed local combination & const.\ in $k$ & \checkmark & \xmark & \xmark \\
Multi-chip ens.~\citep{park2025multichip} & classical aggregation, no cut & n/a & \checkmark & \xmark & \xmark \\
\citet{kawase2024} & param.-free expectation sum & const.\ in $k$ & \checkmark & \xmark & \xmark \\
\textbf{Late fusion (ours)} & \emph{trained} nonlinear head & const.\ in $k$ & \checkmark & \checkmark\ ($Q$-dial) & \checkmark\ ($S_A$, knee)$^{\dagger}$ \\
\bottomrule
\end{tabular}
\end{table*}

\paragraph{Independent quantum modules with classical aggregation.}
Running several small quantum circuits independently and combining their outputs classically
is an established architectural pattern that predates and surrounds circuit cutting.
\citet{tuysuz2023} \emph{classically split} a parametrized circuit by deleting
its inter-partition gates and execute the resulting $O(\log N)$-qubit ansätze independently.
They show that the split mitigates barren plateaus, and they classify classical and quantum
data with a fixed combination of the local outputs. This work is the closest precursor of our
uncoupled subcircuits. \citet{wu2022scalable} use independent quantum devices as
local feature extractors whose outputs feed a downstream predictor, and \citet{kawase2024}
builds on this predecessor. Several further lines of work share the same pattern. Federated
QML~\citep{chen2021federated,chehimi2022qfl} trains quantum models on distributed data with
classical coordination. Scalable and split QCNNs~\citep{baek2022scalable,chinzei2024splitting}
partition convolutional circuits and parallelize them by exploiting translational symmetry.
Multi-circuit feature fusion~\citep{yurtseven2026fusion} and multi-chip
ensembles~\citep{park2025multichip,quantummario2025} partition high-dimensional inputs across
independently operated circuits and aggregate their outputs classically. None of these methods
cuts a circuit and reconstructs it, so none provides a reconstruction baseline. Together they
establish that ``independent modules $+$ classical aggregation'' is \emph{not} our
contribution. Our work instead differs from this pattern in four respects. We place a
\emph{learned} fusion readout \emph{inside a circuit-cutting pipeline} as a drop-in
alternative to quasiprobability reconstruction. We compare the two readouts head-to-head on
the same cut circuit, at matched cost and under noise. We expose the continuum between them as
a tunable reconstruction dial. Lastly, we attach a cut-entanglement diagnostic that indicates
where on the continuum a task lies.

\paragraph{Readouts for cut circuits.}
Three prior readouts for cut circuits are the closest to ours. \citet{marshall2023} cut a
large circuit and learn a \emph{truncated} subset of the reconstruction sum, which is trained
jointly \emph{through} the cut. \citet{kawase2024} partitions input features across small QNNs
and combines them with a parameter-free ensemble sum of expectation values, thereby avoiding
the exponential subexperiment count. \citet{marchisio2024} train parameters across cut
subcircuits by re-encoding mid-circuit measurements, which produces an approximate model. Our
work differs from these methods on three axes. First, we train subcircuits without cross-cut
coupling, so that no gradient flows through the cut. Second, we combine them with a
\emph{trained} classical head rather than a fixed or truncated sum. Third, we expose the full
reconstruction$\leftrightarrow$fusion frontier as a tunable reconstruction dial, together with
a cut-entanglement diagnostic that indicates the operating point. Table~\ref{tab:related}
summarizes these axes. In our experiments, we reimplement the \citet{kawase2024}
parameter-free sum as baseline B3. This baseline also stands in for the classical-splitting
family~\citep{tuysuz2023,wu2022scalable}, whose combiners are likewise fixed functions of
independently measured local outputs. We further compare a Marshall-style
\emph{learned-truncation} readout head-to-head with the $Q$-dial at matched budgets (Results).
Note that the mid-circuit re-encoding of \citet{marchisio2024} is a different model class; a
faithful comparison requires its feedback loop and is left to future work.

\paragraph{Quantum ensembles.}
Quantum ensembles also combine several circuits, but they differ from our setting in what is
combined. Coherent ensembles superpose untrained classifiers~\citep{schuld2018} and are
dequantizable~\citep{abbas2020}, whereas classical-combination ensembles bag independently
trained small VQCs~\citep{li2024,singlequbitensemble2024}. Both families combine
\emph{redundant experts on the same input}. Our subcircuits are instead \emph{partitions of one
decomposed circuit}, so that the correlation that fusion discards is measurable and the uncut
circuit bounds the attainable accuracy.

\paragraph{Multimodal fusion.}
The early/intermediate/late fusion taxonomy~\citep{baltrusaitis2019,snoek2005,atrey2010} is
the source of our combiner. In that literature, late fusion is favored for heterogeneous,
unaligned, independently trained experts, which is exactly the situation of independently
trained subcircuits. Quantum multimodal models fuse \emph{real} modalities through
circuits~\citep{qfl2025}. Our work instead treats the \emph{subcircuits themselves} as
modalities.

\paragraph{Benchmarking QML.}
Our evaluation protocol follows recent benchmarking practice. \citet{bowles2024}
show that classical models often match QML at simulable scales and that removing entanglement
frequently changes nothing. Our work adopts their controls, namely tuned classical baselines
and an entanglement ablation, and reports the regime structure rather than claiming advantage.

\section{Methodology}

This section details our methodology. We first fix the setup and notation and contrast
quasiprobability reconstruction with late fusion as function classes. We then describe
uncoupled late-fusion training, introduce the reconstruction dial $Q$, and define the
cut-entanglement diagnostic. Lastly, we state when late fusion is lossless, namely the
locality condition under which it loses nothing relative to reconstruction.

\subsection{Setup and notation}
A target QNN acts on $n$ qubits that are partitioned into two registers, $A$ and $B$. Cutting
the cross-register gates yields subcircuits that run independently. For an input $x$, whose
features are split across $A$ and $B$, subcircuit $j\in\{A,B\}$ produces a measurement feature
vector $m_j(x)\in\mathbb{R}^{d_j}$, whose entries are the expectations of a chosen set of
local observables. We compare readouts that map the pair $(m_A,m_B)$ to a prediction
$\hat y$.

\paragraph{Reconstruction and fusion as function classes.}
The two readouts under comparison are best contrasted as function classes. Each cross gate is
written in its operator-Schmidt form $G=\sum_\mu g_\mu\,(L_\mu\otimes R_\mu)$. Any product
observable $O_A\otimes O_B$ of the uncut circuit then reconstructs exactly as
\begin{equation}
\langle O\rangle=\!\!\sum_{\boldsymbol\mu,\boldsymbol\nu}\!
\overline{c_{\boldsymbol\nu}}\,c_{\boldsymbol\mu}\,
\langle A_{\boldsymbol\nu}|O_A|A_{\boldsymbol\mu}\rangle
\langle B_{\boldsymbol\nu}|O_B|B_{\boldsymbol\mu}\rangle,
\label{eq:recon}
\end{equation}
with $c_{\boldsymbol\mu}=\prod_j g_{\mu_j}$ and branch states $|A_{\boldsymbol\mu}\rangle$,
$|B_{\boldsymbol\mu}\rangle$. Quasiprobability reconstruction is therefore a \emph{fixed,
linear} readout whose cost is $O(\gamma^{2k})$, where $\gamma=\sum_\mu|g_\mu|$. \emph{Late
fusion} instead computes $\hat y=f_\theta(m_A,m_B)$ with a \emph{trainable classical} head
$f_\theta$, and its readout cost is $n_A{+}n_B$ measurement settings, which is
\emph{independent of} $k$. The price of this saving is that fusion discards the cross-branch
structure, that is, the terms with $\boldsymbol\mu\neq\boldsymbol\nu$, which carries the
cross-cut quantum correlation.

\subsection{Uncoupled (``independent'') late-fusion training}
Uncoupled fusion trains the two subcircuits with \emph{no} connecting gate between them; the
subcircuits are optimized against the shared task loss through the head, as summarized in
Alg.~\ref{alg:fusion}. Throughout the article, \emph{independent} has one precise meaning: no
gradient crosses the cut, and no cross-cut entanglement is ever created. In other words, the
subcircuits are \emph{quantum-mechanically uncoupled}, and they couple only classically,
through $f_\theta$ and the one shared loss. Note that they are consequently not statistically
independent learners. The precise word is ``uncoupled'', and we use the two terms
interchangeably.

\begin{algorithm}[t]
\caption{Independent late-fusion training}
\label{alg:fusion}
\begin{algorithmic}[1]
\REQUIRE data $\{(x_i,y_i)\}$; subcircuit ans\"atze $U_A,U_B$; fusion head $f_\theta$
\STATE initialize subcircuit params $\phi_A,\phi_B$ and head params $\theta$
\REPEAT
  \FOR{each (mini)batch}
    \STATE $m_A^i \gets \textsc{Measure}(U_A(\phi_A);x_i^A)$ \COMMENT{local, no coupling}
    \STATE $m_B^i \gets \textsc{Measure}(U_B(\phi_B);x_i^B)$
    \STATE $\hat y_i \gets f_\theta(m_A^i,m_B^i)$
    \STATE update $\phi_A,\phi_B,\theta$ by $\nabla\,\mathcal{L}(\hat y_i,y_i)$
  \ENDFOR
\UNTIL{converged}
\RETURN $\phi_A,\phi_B,\theta$
\end{algorithmic}
\end{algorithm}

\subsection{The reconstruction dial \texorpdfstring{$Q$}{Q}}
The reconstruction dial interpolates between fusion and reconstruction through a physically
executable budget. Physically, cutting rewrites each cross gate as a quasiprobability
decomposition (QPD) over local operations, so that for a product observable
\begin{equation}
\langle O\rangle=\sum_{i=1}^{6^k} a_i\,E^A_i\,E^B_i,\qquad
\sum_i|a_i|=\kappa^k,
\label{eq:qpd}
\end{equation}
where every $E^A_i,E^B_i$ is the expectation of a \emph{physically executable} subexperiment,
including the parity of its mid-circuit QPD measurements. For an $R_{ZZ}(\phi)$ cut the
constant is $\kappa=1{+}2|\sin\phi|$, which gives $\kappa{=}3$ at $\phi{=}\pi/2$. The
Monte-Carlo cost of estimating a signed sum is the \emph{square} of its coefficient $1$-norm,
namely $\kappa^{2k}$.

The dial truncates this sum by coefficient magnitude. More concretely, to interpolate between
fusion and reconstruction, we keep the largest-$|a_i|$ terms until their $1$-norm reaches a
fraction $Q$ of $\kappa^k$. We estimate only those terms and feed the partial reconstruction
\emph{together with} the raw subcircuit features to the head (Alg.~\ref{alg:qdial}). In turn,
the retained sampling overhead is $(\sum_{\text{kept}}|a_i|)^2$, which ranges from $1$ when no
terms are kept to $\kappa^{2k}$ at full reconstruction. Note that the two endpoints of the
dial require care. At $Q{=}0$ the dial is fusion on the coupled circuit's \emph{frozen}
features; we call this readout \emph{frozen fusion} (setting S1), and it is not the uncoupled
model of Alg.~\ref{alg:fusion}. At $Q{=}1$ the dial is exact reconstruction with a retrained
head. The coefficients and subexperiments come from the trained coupled circuit, which is the
model that reconstruction would use, and the head is retrained for each value of $Q$. The
smallest $Q$ whose accuracy lies within $0.02$ of the $Q{=}1$ accuracy defines the cheapest
sufficient budget $Q^\star$.

We implement the dial through \texttt{qiskit-addon-cutting}: an exact (infinite-shot)
estimator serves the controlled studies, and a finite-shot sampler serves the physically
executed panel. An earlier exact-simulation variant truncates the operator-Schmidt sum of
Eq.~\eqref{eq:recon} and reports the retained Schmidt mass. We use that variant for the
$104$-run diagnostic grid, and we compare it to the physical dial in the Results.

\begin{algorithm}[t]
\caption{$Q$-dial readout (partial physical reconstruction $+$ fusion)}
\label{alg:qdial}
\begin{algorithmic}[1]
\REQUIRE QPD terms $\{(a_i,\,\text{subexp.}^A_i,\,\text{subexp.}^B_i)\}$ of the trained coupled circuit; observables $\{O\}$; budget $Q$
\STATE sort terms by $|a_i|$ descending; keep the top terms until $\sum_{\text{kept}}|a_i|\ge Q\sum_i|a_i|$
\STATE run only the kept subexperiments; $r_O \gets \sum_{\text{kept}} a_i E^A_i E^B_i$, $\forall O$ \COMMENT{cost $(\sum_{\text{kept}}|a_i|)^2$}
\STATE measure $m_A,m_B$ locally; fit $f_\theta$ on $[\,\{r_O\}\ \|\ m_A\ \|\ m_B\,]$
\RETURN $f_\theta$ and the retained overhead
\end{algorithmic}
\end{algorithm}

\subsection{Cut-entanglement diagnostic}
The cut-entanglement diagnostic quantifies how much cross-cut correlation the trained circuit
creates. To this end, we measure the entanglement entropy $S_A(x)=S(\rho_A(x))$ across the cut
of the trained circuit and average it over the data. By the bound of \citet{jing2024}, the
reconstruction overhead is exponentially lower-bounded by this quantity. In other words, the
resource that fusion \emph{saves} and the correlation that it \emph{drops} are one and the
same, which makes $S_A$ a natural predictor of the $Q$-dial knee and an indicator of the
fusion accuracy gap. We stress that $S_A$ is a property of the \emph{trained} coupled circuit.
At initialization, or after only five optimizer iterations, it carries no signal: the Spearman
correlation satisfies $|\rho|{<}0.1$ over $104$ runs. The value of the diagnostic lies in
predicting the knee without paying the $\gamma^{2k}$ sampling cost of a full dial sweep.

\subsection{When is late fusion lossless?}
We now state the condition under which fusion loses nothing relative to reconstruction, and
its converse.
\begin{proposition}[Locality condition]
\label{prop:lossless}
Write $M=(m_A,m_B)$ for the local features, $\eta_M=\Pr[y{=}1\mid M]$, and, for a cross-cut
correlator $C=\langle O_A\otimes O_B\rangle$ not determined by $(\rho_A,\rho_B)$,
$\eta_{MC}=\Pr[y{=}1\mid M,C]$.
\textbf{(i) Losslessness.} If the Bayes-optimal predictor depends on the joint state
$\rho_{AB}(x)$ only through the marginals $\rho_A(x),\rho_B(x)$---equivalently, if $M$ is a
sufficient statistic for $y$---then a sufficiently expressive fusion head $f_\theta(M)$ attains
the Bayes risk and reconstruction cannot improve on it.
\textbf{(ii) Converse.} The excess risk of the best fusion predictor over the best joint
predictor equals
$\mathbb{E}_M\big[\min(\eta_M,1{-}\eta_M)\big]-\mathbb{E}\big[\min(\eta_{MC},1{-}\eta_{MC})\big]\ge0$,
with strict inequality if and only if the sign of $\eta_{MC}-\tfrac12$ varies with $C$ on a set
of positive probability (a \emph{decision-relevant} correlator). In particular, if the label
is determined by $(M,C)$, $H(y\mid M,C)=0$, while $C$ carries residual information
$\varepsilon=I(y;C\mid M)>0$, every fusion predictor has error at least $h_2^{-1}(\varepsilon)$
whereas joint readout attains zero error.
\end{proposition}
\noindent\emph{Proof sketch.} (i) Thresholding $\eta_M$ attains the Bayes risk and is a function
of $M$ alone, which a universal approximator represents. (ii) Given $M$, the Bayes risk without
$C$ is $\min(\eta_M,1{-}\eta_M)$ and with $C$ is $\mathbb{E}[\min(\eta_{MC},1{-}\eta_{MC})\mid M]$;
since $\eta_M=\mathbb{E}[\eta_{MC}\mid M]$ and $\min(t,1{-}t)$ is concave, Jensen gives the
inequality, with equality iff $\eta_{MC}-\tfrac12$ has constant sign a.s.\ given $M$. Under
$H(y\mid M,C)=0$ the chain rule gives $H(y\mid M)=\varepsilon$ and Fano's inequality yields the
bound. Full proofs in the appendix. $\square$

\noindent\emph{Remark.} Positive conditional information $I(y;C\mid M)>0$ alone does
\emph{not} force an accuracy loss: the correlator can move the posterior without moving it
across $\tfrac12$, leaving the Bayes decision unchanged. The decision-relevance condition in
(ii) is the exact criterion; the information bound applies when the label is determined by
$(M,C)$, which the entangled-data boundary experiment realizes.

\noindent Proposition~\ref{prop:lossless} frames the empirical regime map reported later in
the article. Put briefly, fusion is lossless when the task's label information is local, and
it must lose accuracy when a cross-cut correlator is decision-relevant. The diagnostic $S_A$
introduced above is an \emph{empirical} indicator of departure from the locality condition,
and we evaluate it in the Results. Note that the proposition itself does not bound the gap in
terms of $S_A$.

\section{Experimental Setup}\label{sec:setup}
This section details the setup of our study. We first define the three fusion settings and
the two cost conventions that the results use, then describe the circuits and cutting
engines, the baselines, the datasets and observable sets, and the cost metrics, and lastly
the evaluation metrics and protocol.

\subsection{Fusion settings and cost conventions}
Three distinct fusion settings appear in the results, and each Results subsection names the
setting it uses. In \textbf{(S1) frozen fusion}, we train a coupled QNN in the usual way,
execute its two subcircuits \emph{without} the coupling gates, and fit a head on the
resulting local features. In \textbf{(S2) uncoupled fusion} (Alg.~\ref{alg:fusion}), we
train subcircuits that contain no coupling gate from scratch together with the head. In
\textbf{(S3) the reconstruction dial}, we retain the trained coupled QNN of (S1) and fit a
head on partial-reconstruction features together with the local features, for each retained
budget $Q$. Note that the dial endpoint $Q{=}0$ therefore coincides with (S1) and not with
(S2). Two sampling-overhead conventions also appear. The exact-simulation convention reports
$\gamma_S^{2k}$, where $\gamma_S$ is the operator-Schmidt $1$-norm of the \emph{trained}
coupling gate; we use this convention in the controlled statevector studies. The physical
quasiprobability convention reports $(\sum_i|a_i|)^2$, which equals $9^k$ for $R_{ZZ}(\pi/2)$
cuts and which we verified through \texttt{qiskit-addon-cutting}. Table~\ref{tab:settings}
summarizes the settings and lists where each one is used.

\begin{table*}[t]
\caption{The fusion settings compared in this paper. ``Fixed head'': the coupled QNN's own
trained linear readout.}
\label{tab:settings}
\centering\footnotesize
\setlength{\tabcolsep}{3pt}
\begin{tabular}{@{}llll@{}}
\toprule
setting & subcircuits trained & head input & used in \\
\midrule
(S1) frozen fusion & inside coupled QNN & local $m_A,m_B$ of (S1) circuit & Fig.~\ref{fig:overhead}(a), Fig.~\ref{fig:independent}, dial $Q{=}0$ \\
(S2) uncoupled fusion & no coupling, with head & local $m_A,m_B$ & Figs.~\ref{fig:independent}, \ref{fig:noise}, \ref{fig:scaling}; Tables \ref{tab:bench}, \ref{tab:faithful} \\
(S3) dial, $0{<}Q{\le}1$ & inside coupled QNN & partial recon.\ $+$ local & Fig.~\ref{fig:qdial}, diagnostic study \\
reconstruction (B1) & inside coupled QNN & exact recon., fixed head & every comparison \\
\bottomrule
\end{tabular}
\end{table*}

Figure~\ref{fig:pipeline} depicts the complete tool flow. We list the approximate compute per
study in the appendix.

\begin{figure*}[t]
\centering
\includegraphics[width=\widefigwidth]{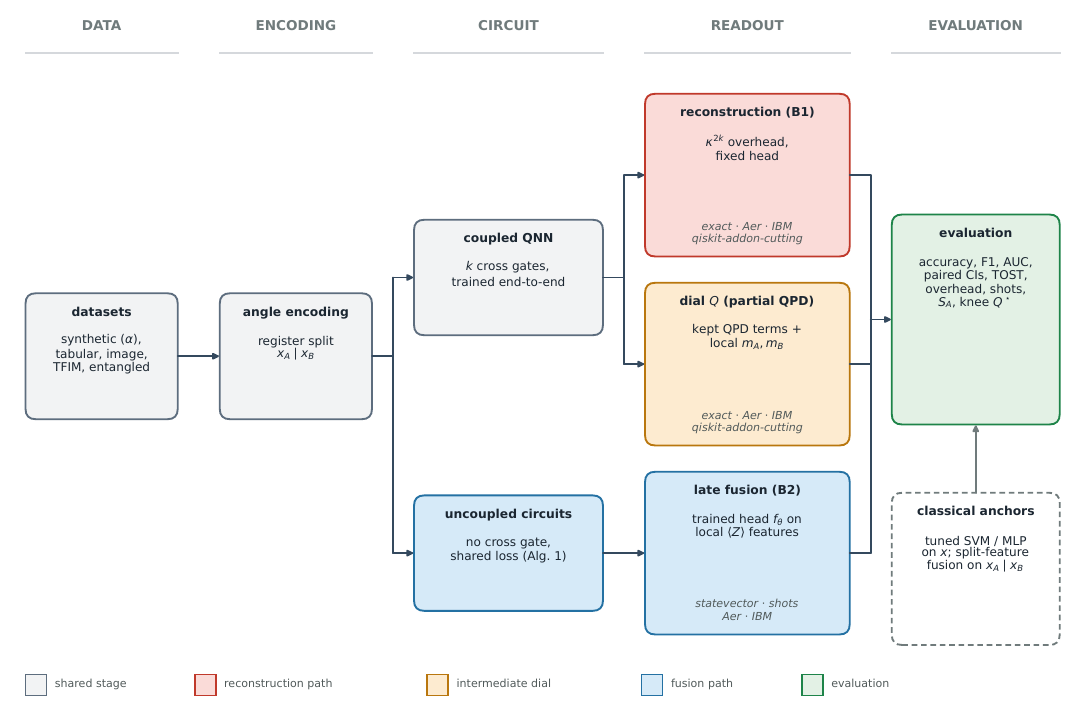}
\caption{Experimental and tool flow. Column headers name the stages. Every dataset is
angle-encoded and split across two registers. The coupled QNN feeds the reconstruction readout
(B1) and the dial, whereas the uncoupled subcircuits feed the late-fusion head (B2). Box colour
marks the path: grey for shared stages, red for reconstruction, amber for the intermediate dial,
blue for fusion and green for evaluation. The italic lines inside each readout box name the
engines on which it was executed (exact statevector, \texttt{qiskit-addon-cutting}, Aer noise
models, IBM Heron devices). The classical anchors (dashed box) enter the evaluation only.}
\label{fig:pipeline}
\end{figure*}

\subsection{Circuits and cutting}
The target QNNs use a hardware-efficient ansatz on two registers $A$ and $B$ that are joined by
$k$ cross gates. Each register carries per-qubit $R_Y,R_Z$ rotations followed by a $CZ$
entangling ring; we refer to this family as ansatz~A. To check that our findings do not depend
on this particular gate set, we repeat the central comparisons with a second family, ansatz~B.
In ansatz~B each register carries $R_X,R_Y$ rotations followed by a CNOT ladder, and the cross
gates are fixed at $\phi{=}\pi/2$, the point of maximal reconstruction overhead ($9$ per cut).
Unless stated otherwise, results refer to ansatz~A. We used two simulation engines. The first is an exact state-vector simulator
with a hand-verified operator-Schmidt gate cut; its reconstruction reproduces the uncut circuit
to $<10^{-10}$, and we used it for the controlled studies. The second is the
\texttt{qiskit-addon-cutting} library with the physical quasiprobability decomposition; we
used it to confirm the results at the level of real subexperiments and $\gamma$ factors.

\subsection{Baselines}
We compare five quantum readouts and several classical anchors throughout. \textbf{B0} is
the uncut circuit and serves as the gold standard. \textbf{B1} is full quasiprobability
reconstruction, which is exact and therefore equals B0. \textbf{B2} is late fusion with a
linear head or a multilayer perceptron (MLP) head. \textbf{B3} is the Kawase-style
parameter-free sum of expectation values~\citep{kawase2024}. \textbf{B4} is the entanglement
ablation, in which the coupling angle is sent to $0$. The classical anchors are tuned support
vector machines with a radial-basis kernel (SVM-RBF) and MLP models trained on the identical
preprocessed features. We also include a \emph{classical split-feature fusion} baseline,
which places one classical encoder per register on the same feature partition, with $n_j$
outputs each, and fuses them by a head of the same capacity as $f_\theta$. This baseline
isolates what the register partition itself costs. For the quantum fusion readout we report
both frozen fusion (S1), whose head is trained on the subcircuit features of a fixed circuit,
and uncoupled fusion (S2), which follows Alg.~\ref{alg:fusion}.

\subsection{Datasets}
The datasets span a controlled synthetic family, standard tabular and two-dimensional sets,
downscaled image tasks, multiclass tasks, and two quantum-native tasks. The synthetic
generator interpolates from \textsc{separable} to \textsc{parity} through a knob $\alpha$ that
tunes the cross-cut label dependence. We draw inputs as $x\sim\mathrm{U}[-1,1]^{n}$. The
scalars $u$ and $v$ are standardized linear projections of the two register halves,
$u\propto w_A\!\cdot\!x_A$ and $v\propto w_B\!\cdot\!x_B$, with fixed random $w_A,w_B$. The
label is $y=\mathbb{1}\big[(1{-}\alpha)(u{+}v)+2\alpha\,uv>0\big]$. At $\alpha{=}0$ the label
is additively separable across the cut. At $\alpha{=}1$ it is the sign parity of the two
halves, and each half alone is at chance. The standard sets are Iris~\citep{fisher1936},
Breast-Cancer~\citep{uci}, Moons, and Circles. The image tasks are binary
Fashion-MNIST~\citep{fashionmnist2017} and MNIST~\citep{mnist1998}, reduced by principal
component analysis to four components (PCA-4) and by PCA-8. The multiclass tasks are 3-class
Iris and MNIST 0/1/2, read by a softmax fusion head. The first quantum-native task is an
entangled-state task with a data-entanglement knob $\theta$: the two classes are
$|\psi_{y=1}\rangle=\cos\theta|00\rangle+\sin\theta|11\rangle$ and
$|\psi_{y=0}\rangle=\cos\theta|01\rangle+\sin\theta|10\rangle$, with small random local
$R_Y$ noise added. Its label is the cross-cut correlator $\langle Z_aZ_b\rangle=\pm1$, and at
$\theta{=}\pi/4$ every local marginal vanishes for both classes. The second quantum-native
task is transverse-field Ising model (TFIM) ground-state classification, for the phase and
for the bond, with ground states obtained by exact diagonalization.

\paragraph{Observable sets.}
The coupled QNN's readout (B0/B1) consists of $\langle Z_q\rangle$ on every qubit and
$\langle Z_aZ_b\rangle$ on every cut pair, read by a trained linear head. The fusion features
$m_A,m_B$ are $\langle Z_q\rangle$ on every qubit of each subcircuit run alone, which yields
$n_A{+}n_B$ numbers. The boundary experiment instead uses all of $\langle X\rangle,\langle
Y\rangle,\langle Z\rangle$ per qubit, and its joint readout additionally uses $\langle
XX\rangle,\langle YY\rangle,\langle ZZ\rangle$. We scale classical features to
$[-\pi/2,\pi/2]$ for angle encoding.

\subsection{Cost metrics}
Two sampling-overhead conventions appear in the results, and we keep them distinct. In the
controlled statevector studies, the reported reconstruction overhead is $\gamma_S^{2k}$, where
$\gamma_S$ is the per-cut operator-Schmidt coefficient $1$-norm of the \emph{trained} coupling
gate; this is an exact-simulation convention. The physical quasiprobability overhead follows a
different per-cut constant, $(1+2|\sin\phi|)^{2}$ for $R_{ZZ}(\phi)$ cuts, which equals
exactly $9$ at $\phi{=}\pi/2$. We verified this constant through
\texttt{qiskit-addon-cutting} (Fig.~\ref{fig:overhead}(b)). Both conventions are exponential
in $k$, and neither bounds the other at general trained angles. The corresponding cost of late
fusion is its measurement-setting count, namely $n_A{+}n_B$ local settings, which is constant
in $k$. The sampling overhead of the physical dial is the squared retained QPD $1$-norm
$(\sum_{\text{kept}}|a_i|)^2$, floored at $1$. The exact-simulation dial of the diagnostic
grid reports the retained Schmidt mass instead, and we show both.

\subsection{Metrics and evaluation protocol}
We use accuracy as the primary metric. We report the F1 score and the area under the receiver
operating characteristic curve (AUC) for the standard datasets in Table~\ref{tab:bench} and
for all runs in the released result files. The resource metrics are the subexperiment count,
the sampling overhead, and the number of shots. Paired per-seed statistics for every headline
comparison appear in the appendix; they comprise bootstrap $95\%$ confidence
intervals, Wilcoxon signed-rank tests, and two one-sided equivalence tests (TOST) at a
predefined margin of $\delta{=}0.02$ accuracy. Throughout this paper, ``matches'' means
equivalence at that margin, and we write ``inconclusive'' or ``deficit'' otherwise.

We report the mean\,$\pm$\,std over seeds. The number of seeds is $20$ for the $\alpha$-sweep,
benchmark, boundary, shots, cost-crossover, harder-data, and Marshall studies; $8$ for the
dense-$\alpha$ diagnostic; and $5$ for the multiclass, ablation, width, and scale-up studies.
Each seed has its own held-out test split, of $30$--$40\%$ depending on the study. For the
real datasets, we fit standardization, PCA, and angle scaling on the \emph{training split
only} and applied them unchanged to the test split, with test features clipped to the
encoding range; this protocol is identical for every model, including the classical anchors.
The synthetic generator produces features i.i.d.\ uniform in $[-1,1]$, which are rescaled to
angles by a label-independent map. In the controlled studies we scaled the number of cuts
$k\!\in\!\{1,2,3\}$ and the total width $n\!\in\!\{4,6,8\}$ independently; the scale-up
study extends these to $n{=}20$ and $k{=}10$. The noise ladder has three rungs. The first rung
is exact state-vector simulation. The second rung is finite shots, modeled as Gaussian noise on
expectation values with the $\gamma$-scaled standard deviation that the quasiprobability
estimator incurs. The third rung is Aer device noise under \emph{two} models, namely a custom
depolarizing-plus-readout model and a real IBM fake-backend calibration
(\texttt{FakeManilaV2}).

\section{Results}\label{sec:results}
This section presents the findings of our study in six parts. We first compare the accuracy
and cost of late fusion against reconstruction. We then locate the regimes in which local
features lose information. Next, we examine both readouts under shot and device noise and on
live hardware, and we follow them as width, cut count and coupling depth grow. We then execute
the reconstruction dial and test the cut-entanglement diagnostic. Lastly, we compare against
prior-art readouts. Each subsection states which fusion setting of Table~\ref{tab:settings}
it uses.

\subsection{Accuracy and cost}
We first examine whether late fusion reaches the accuracy of quasiprobability
reconstruction on standard and harder datasets, and what each readout pays in sampling
overhead. The benchmark and harder-data tables use uncoupled fusion (S2), the cost crossover
uses frozen fusion (S1), and the physical-library check involves no trained head.

\subsubsection{Late fusion tracks reconstruction at exponentially lower sampling overhead}
We begin with the central comparison: the accuracy of late fusion against quasiprobability
reconstruction on standard datasets, and the sampling overhead that each requires.

\paragraph{Late fusion is equivalent to reconstruction on three of four standard datasets.}
Table~\ref{tab:bench} lists the readout ladder on the standard datasets under uncoupled
fusion (S2). We test ``matches'' as \emph{equivalence within a predefined margin} of
$\delta{=}0.02$ accuracy, using two one-sided tests (TOST; appendix D). Under this test, B2
(MLP fusion) is equivalent to B0/B1 on Circles, Iris, and Breast-Cancer. On Moons the result
is inconclusive at $10$ seeds: the gap is $+0.002$, but the confidence interval is wide. We can
also see from the table that the value of the \emph{trained} head over the parameter-free
Kawase sum (B3) and over the entanglement ablation (B4) is largest where feature interactions
matter; on Breast-Cancer the margins are $+0.26$ and $+0.08$, respectively.

\paragraph{Reconstruction overhead grows exponentially with cuts while fusion stays flat.}
Figure~\ref{fig:overhead}(a) depicts the cost crossover over $10$ seeds under frozen fusion
(S1). We can observe that the reconstruction overhead grows exponentially with the number of
cuts ($7.3\to45.9\to408$ for $k=1,2,3$) and that the subexperiment count grows from $4$ to
$16$ to $64$, while the cost of fusion stays flat. On this moderately cross-cut task
($\alpha{=}0.5$), \emph{frozen} fusion trails the uncut model by $0.04$--$0.11$. The next
subsection shows that uncoupled training closes this gap, which is $\le0.03$ at scale
(Fig.~\ref{fig:scaling}(a)). These results show that late fusion is statistically equivalent
to reconstruction on three of the four standard datasets, while its measurement cost remains
constant in the number of cuts.

\begin{table}[t]
\caption{Standard datasets (accuracy, mean$\pm$std, $10$ seeds; F1 and AUC rows for the two
non-saturated sets, std shown where it exceeds $0.005$). Late fusion (B2) stays within $0.01$
of reconstruction (B0/B1 coincide) on every dataset (Circles: $-0.009$, small but significant;
Appendix). It exceeds the Kawase sum (B3) by a large margin and the ablation (B4) by a smaller
margin on Breast-Cancer, with smaller or no margins on the easier sets. The Kawase sum's AUC
on Breast-Cancer ($0.54$) shows that its fixed unweighted sum discards the ranking
information that the trained head recovers.}
\label{tab:bench}
\centering\footnotesize
\setlength{\tabcolsep}{2pt}
\begin{tabular}{@{}lcccc@{}}
\toprule
dataset & B0/B1 & B2\,(MLP) & B3\,(Kawase) & B4\,(abl.) \\
\midrule
Moons & $.879{\pm}.035$ & $.881{\pm}.060$ & $.814{\pm}.048$ & $.843{\pm}.032$ \\
Circles & $.981{\pm}.007$ & $.972{\pm}.012$ & $.974{\pm}.013$ & $.978{\pm}.010$ \\
Iris & $1.00{\pm}.000$ & $1.00{\pm}.000$ & $.937{\pm}.151$ & $1.00{\pm}.000$ \\
Breast-C. & $.939{\pm}.013$ & $.932{\pm}.014$ & $.670{\pm}.063$ & $.849{\pm}.079$ \\
\midrule
Breast-C.\ F1 & $.951$ & $.946{\pm}.011$ & $.788$ & $.874$ \\
Breast-C.\ AUC & $.989$ & $.985{\pm}.006$ & $.539{\pm}.123$ & $.980$ \\
Moons AUC & $.955$ & $.961{\pm}.024$ & $.869{\pm}.060$ & $.935$ \\
\bottomrule
\end{tabular}
\end{table}

\begin{figure*}[t]
\centering
\includegraphics[width=\panelfigwidth]{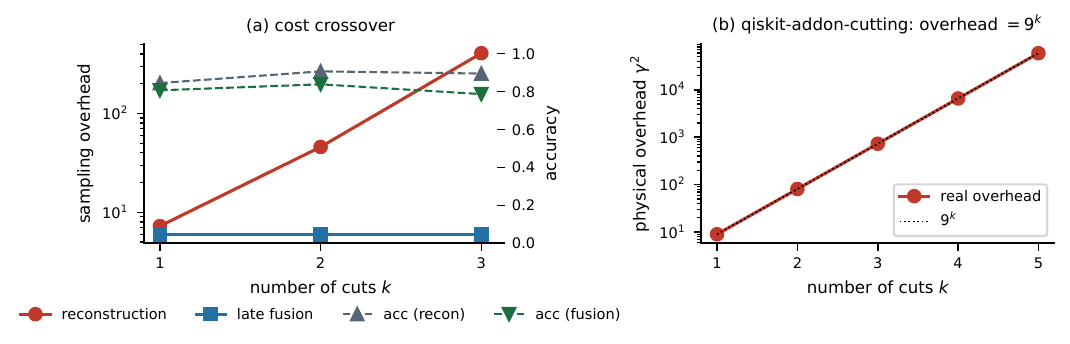}
\caption{(a) Cost crossover on the controlled simulator: reconstruction overhead (red) grows
exponentially with the number of cuts while late fusion (blue) stays flat; reconstruction
accuracy (gray dashed) tracks the ideal circuit, frozen fusion (green dashed) trails by
$0.04$--$0.11$ (see text). (b) The same wall through the real \texttt{qiskit-addon-cutting}
library: measured physical overhead is exactly $9^k$ (dotted), with reconstruction matching
exact expectations to shot noise.}
\label{fig:overhead}
\end{figure*}

\subsubsection{Uncoupled training closes the residual gap}
Frozen fusion (S1) leaves a residual gap at high cross-cut dependence, and we aim to shed
light on whether the gap disappears once the subcircuits are trained for the fusion
objective, that is, under uncoupled fusion (S2).

\paragraph{Uncoupled training reduces the frozen-fusion shortfall from \texorpdfstring{$0.111$}{0.111} to \texorpdfstring{$0.010$}{0.010}.}
Figure~\ref{fig:independent} depicts both fusion settings against reconstruction across the
cross-cut sweep. Re-using frozen subcircuits leaves a residual gap at high cross-cut
dependence, and training the subcircuits \emph{for} the fusion objective
(Alg.~\ref{alg:fusion}) closes it. Over the $\alpha$ sweep ($20$ seeds per cell), the mean shortfall relative to
reconstruction falls from $0.111$ under frozen fusion to $0.010$ under uncoupled fusion. The
pooled paired gap (fusion$-$reconstruction) is $-0.010$, with a $95\%$ confidence interval (CI) of
$[-0.019,-0.001]$ over $80$ seed-pairs: a small offset that is statistically detectable but
lies inside the margin $\delta{=}0.02$, so the pooled comparison is equivalent by the TOST.

\paragraph{Small but genuine deficits remain on the most cross-cut tasks.}
The individual $20$-seed cells separate the regimes. At $\alpha{=}0.25$ and $\alpha{=}0.5$
uncoupled fusion is equivalent to reconstruction (gaps $+0.008$ and $+0.003$, CIs inside
$\pm\delta$). At $\alpha{=}0.75$ the cell is inconclusive (gap $-0.015$, CI
$[-0.036,+0.003]$). At pure parity the cell shows a significant \emph{deficit} of $-0.035$,
with CI $[-0.055,-0.016]$ and $p{=}0.006$. The $\alpha{=}0.75$ task of Table~\ref{tab:faithful} also shows a deficit, of
$-0.024$. Fusion therefore tracks reconstruction closely, with small but genuine losses on the
most cross-cut tasks, and we do not claim equivalence there.

\paragraph{The trained subcircuits, not the head alone, carry the signal.}
Towards understanding where the signal comes from, we fitted the same head on
\emph{untrained} random-parameter subcircuits. It reaches only $0.51$--$0.64$ across the
sweep. On easy tabular data, by contrast, random features suffice and reach $0.93$ on
Breast-Cancer. These results show that uncoupled training reduces the pooled frozen-fusion
shortfall from $0.111$ to $0.010$ and leaves only small deficits at the highest cross-cut
dependence.

\begin{figure}[t]
\centering
\includegraphics[width=\singlefigwidth]{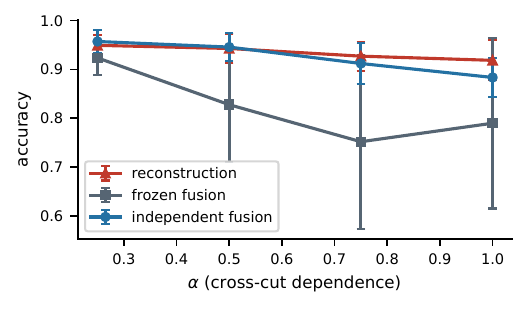}
\caption{Uncoupled training (circles) matches reconstruction (triangles) across the cross-cut
sweep, closing the residual gap left by frozen fusion (squares).}
\label{fig:independent}
\end{figure}

\subsubsection{Harder data and faithful prior-art baselines}
We next test whether the same pattern holds on harder image data and against faithful
re-implementations of prior readouts, again with uncoupled fusion (S2).

\paragraph{Fusion matches reconstruction and the classical anchor on image data.}
Table~\ref{tab:faithful} consolidates the harder-data comparison: \CUTfaithful{} On the local
image tasks the \citet{kawase2024} sum ties fusion. The trained head is decisive on
cross-cut structure, where the fixed sum collapses. A softmax head extends the method to
$3$-class tasks within ${\sim}0.01$ of the classical anchor (appendix). On the image tasks,
then, fusion matches reconstruction and the classical anchor, and it retains a small but
statistically detectable shortfall only on the hard synthetic task.

\begin{table}[t]
\caption{Harder data and faithful baselines (mean$\pm$std, $10$ seeds; ``-8'' rows are PCA-8
on $4{+}4$ qubit registers, others PCA-4 on $2{+}2$).}
\label{tab:faithful}
\centering\footnotesize
\setlength{\tabcolsep}{1.5pt}
\begin{tabular}{@{}lcccc@{}}
\toprule
dataset & recon.\,(B1) & \textbf{fusion\,(B2)} & Kawase\,(B3) & classical \\
\midrule
F-MNIST & $.799{\pm}.054$ & $.799{\pm}.054$ & $.790{\pm}.059$ & $.797{\pm}.044$ \\
MNIST & $.988{\pm}.009$ & $.990{\pm}.008$ & $.982{\pm}.009$ & $.991{\pm}.009$ \\
F-MNIST-8 & $.796{\pm}.052$ & $.786{\pm}.051$ & $.797{\pm}.032$ & $.791{\pm}.041$ \\
MNIST-8 & $.985{\pm}.017$ & $.987{\pm}.014$ & $.981{\pm}.014$ & $.990{\pm}.010$ \\
synth.\ $\alpha{=}0.75$ & $.931{\pm}.026$ & $.907{\pm}.038$ & $.564{\pm}.089$ & $.940{\pm}.027$ \\
\bottomrule
\end{tabular}
\end{table}

\subsubsection{Robustness to the circuit family}
All results so far use ansatz~A. One may reasonably ask whether the match between fusion and
reconstruction is a property of that particular gate set. We therefore repeated the cross-cut
sweep and the four standard benchmarks with ansatz~B (Section~\ref{sec:setup}), whose cross
gates sit at the maximal-overhead point, and we applied the same statistical protocol.

\paragraph{The fusion result holds under a second circuit family.}
Table~\ref{tab:ansatzB} places the two families side by side. Under ansatz~B the uncoupled
readout pools to a gap of $-0.001$ [$-0.014$, $+0.011$] over $40$ pairs and passes the
equivalence test at $\delta{=}0.02$, as ansatz~A does over $80$ pairs. Per $\alpha$, the
$10$-seed cells of ansatz~B are equivalent at $\alpha{=}0.5$ and inconclusive elsewhere, with
intervals that enclose zero in every case. The parity deficit observed under ansatz~A
($-0.035$) does not reappear under ansatz~B ($-0.010$ [$-0.034$, $+0.017$]), so the size of
the residual loss at $\alpha{=}1$ depends on the circuit family, whereas the match at low and
moderate cross-cut dependence does not. Frozen fusion loses $0.07$--$0.20$ under ansatz~B, as
it does under ansatz~A, which confirms that uncoupled training is what closes the gap. The
trained ansatz-B circuits entangle the cut strongly, at $0.87$--$0.95$ bits across the sweep,
and the uncoupled readout matches reconstruction nonetheless, in line with the interpretation
in Section~\ref{sec:lost} that the entanglement a model builds is not the information the
task requires. On the four standard datasets the pattern repeats, with one difference: Circles
shows a deficit of $-0.016$ [$-0.027$, $-0.007$] under ansatz~B, which sits just outside the
margin, where ansatz~A showed a small offset of $-0.009$ inside it. Breast-Cancer, Iris and
Moons remain equivalent or inconclusive under both families. We conclude that the qualitative
result, a cut-independent readout that matches reconstruction on local tasks, is not tied to
one gate set, while the exact size of the residual deficit on the hardest tasks is.

\ifqmi\begin{table}[t]\else\begin{table*}[t]\fi
\caption{Uncoupled fusion minus reconstruction (B2$-$B0) under the two circuit families:
paired-bootstrap mean and $95\%$ CI, with the TOST verdict at $\delta{=}0.02$ (equivalent,
inconclusive, or deficit). Sweep cells use $20$ seeds for ansatz~A and $10$ for ansatz~B;
benchmark cells use $10$ seeds for both. $^\dagger$Equivalent with a small but significant
offset inside the margin.}
\label{tab:ansatzB}
\centering\footnotesize
\setlength{\tabcolsep}{2.5pt}
\begin{tabular}{@{}lcccc@{}}
\toprule
 & \multicolumn{2}{c}{ansatz A} & \multicolumn{2}{c}{ansatz B} \\
\cmidrule(lr){2-3}\cmidrule(lr){4-5}
condition & gap [95\% CI] & verdict & gap [95\% CI] & verdict \\
\midrule
$\alpha{=}0.25$ & $+.008$ [$-.004$, $+.019$] & equiv. & $+.014$ [$-.003$, $+.031$] & incon. \\
$\alpha{=}0.5$  & $+.003$ [$-.011$, $+.017$] & equiv. & $+.003$ [$-.011$, $+.016$] & equiv. \\
$\alpha{=}0.75$ & $-.015$ [$-.036$, $+.003$] & incon. & $-.013$ [$-.050$, $+.021$] & incon. \\
$\alpha{=}1.0$  & $-.035$ [$-.055$, $-.016$] & deficit & $-.010$ [$-.034$, $+.017$] & incon. \\
pooled sweep    & $-.010$ [$-.019$, $-.001$] & equiv.$^\dagger$ & $-.001$ [$-.014$, $+.011$] & equiv. \\
\midrule
Moons         & $+.002$ [$-.028$, $+.038$] & incon. & $-.004$ [$-.036$, $+.030$] & incon. \\
Circles       & $-.009$ [$-.016$, $-.003$] & equiv.$^\dagger$ & $-.016$ [$-.027$, $-.007$] & deficit \\
Iris          & $.000$ [$.000$, $.000$] & equiv. & $.000$ [$.000$, $.000$] & equiv. \\
Breast-Cancer & $-.006$ [$-.015$, $+.001$] & equiv. & $+.003$ [$+.000$, $+.006$] & equiv.$^\dagger$ \\
\bottomrule
\end{tabular}
\ifqmi\end{table}\else\end{table*}\fi

\subsubsection{Confirmation with the physical cutting library}
The overhead growth observed on the controlled simulator could in principle be an artifact of
our implementation, so we repeat the measurement with the physical quasiprobability
decomposition. The comparison is between the estimation error of one reconstructed expectation
value and that of one local fusion feature, and it involves no trained head.

\paragraph{The physical library measures an overhead of exactly \texorpdfstring{$9^k$}{9\textasciicircum k}.}
To exclude a simulator artifact, we re-ran the pipeline through the
\texttt{qiskit-addon-cutting} quasiprobability library with an Aer sampler, up to $k{=}5$
cuts. Figure~\ref{fig:overhead}(b) shows that the measured sampling overhead is
\emph{exactly} $9^k$, reaching $59{,}049$ at $k{=}5$.

\paragraph{The exponential cost appears in the error budget.}
More closely, we see that the ${\sim}6^k$ subexperiments force the per-subexperiment shots
down from $10^5$ to $4{\times}10^3$, and that the reconstruction error rises from $0.002$ to
$0.068$ between $k{=}3$ and $k{=}5$. Roughly $6{\times}10^7$ total shots are thus spent for
an error of $0.07$, whereas a single fusion feature reaches an error of ${\sim}0.01$ from
$10^4$ shots. This shows that the overhead growth that fusion avoids is a property of the
physical decomposition and not of our implementation.

\paragraph{Summary.}
Overall, our findings suggest that late fusion with a trained head reaches the accuracy of
quasiprobability reconstruction on standard and image data at a measurement cost that is
constant in the number of cuts, whereas reconstruction pays an overhead of $9^k$ in the
physical basis. Uncoupled training supplies most of this result: it re-encodes the locally
available information and leaves only small deficits where the label depends most strongly on
the cross-cut interaction. The result holds under a second circuit family, which indicates
that it reflects the readout rather than the gate set.

\subsection{Where information is lost}\label{sec:lost}
The preceding subsection established that fusion matches reconstruction on classical data at
constant measurement cost. We now turn to the regimes in which local features are
insufficient: entangled quantum data, a quantum-native many-body task, and classical inputs
split across registers.

\subsubsection{The quantum-necessity boundary}
Proposition~\ref{prop:lossless} predicts that late fusion must fail when the label resides in
a cross-cut correlator, and we construct that case here. The fusion readout is a head on the
local observables of the input state (the boundary observable set of the Experimental Setup);
no coupled circuit is trained, so the distinction between (S1) and (S2) does not arise.

\paragraph{Fusion collapses to chance when the label lives in the cross-cut correlator.}
Figure~\ref{fig:regime} depicts the quantum-necessity boundary over $10$ seeds. On entangled
data, the label is carried by a cross-cut correlator, and the local marginals are
uninformative. We can observe that late fusion holds $\ge\!0.97$ accuracy through $0.988$
bits of data entanglement, falls to $0.78$ at $0.998$ bits, and reaches chance ($0.51$) at
maximal entanglement, whereas reconstruction stays $\ge\!0.994$ throughout. This is exactly
the regime that Proposition~\ref{prop:lossless} predicts fusion cannot serve: fusion
collapses to chance precisely when the label information resides entirely in the cross-cut
correlation.

\begin{figure}[t]
\centering
\includegraphics[width=\singlefigwidth]{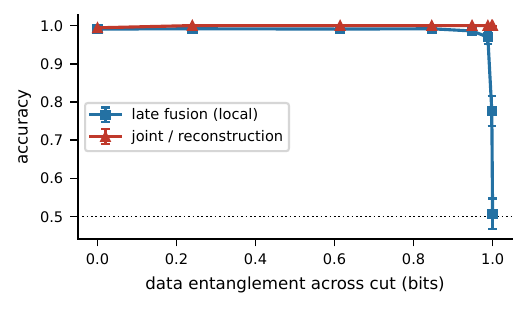}
\caption{Quantum-necessity regime: as data entanglement across the cut grows, late fusion
(local measurements) collapses to chance while joint/reconstruction readout stays at or
near perfect ($\ge\!0.994$).}
\label{fig:regime}
\end{figure}

\subsubsection{A quantum-native task: TFIM}
A quantum-state input raises a different question, namely which quantum readout to pay for
rather than whether one is needed. As in the boundary experiment, the fusion readout is a head
on local observables of the state, and no coupled circuit is trained. For
\emph{quantum-state} inputs, classical machine learning on raw features is undefined, so the
only question is which quantum readout to pay for. To this end, we used transverse-field
Ising ground states with $n{=}8$ spins, obtained by exact diagonalization, with $400$-shot
features.

\paragraph{Fusion is lossless on phase classification and nearly lossless on bond detection.}
On phase classification, fusion matches joint readout: both reach $1.000$ over $5$ seeds,
because the order parameters are local. Bond detection asks whether the cross-cut coupling is
physically present. The cut-entanglement diagnostic separates this label exactly, with
$S_A{\approx}0.15$ against $0$. On this task fusion reaches $0.989$, against $1.000$ for
joint readout. The small loss arises because deleting the bond also shifts the boundary-spin
marginals; cross-cut dependence therefore need not be locally invisible. The diagnostic
$S_A(g)$ peaks near criticality (Fig.~\ref{fig:scaling}(c)). On this quantum-native task,
then, fusion is lossless when the order parameter is local and nearly lossless even when the
label concerns the cut itself.

\begin{figure*}[!t]
\centering
\includegraphics[width=\panelfigwidth]{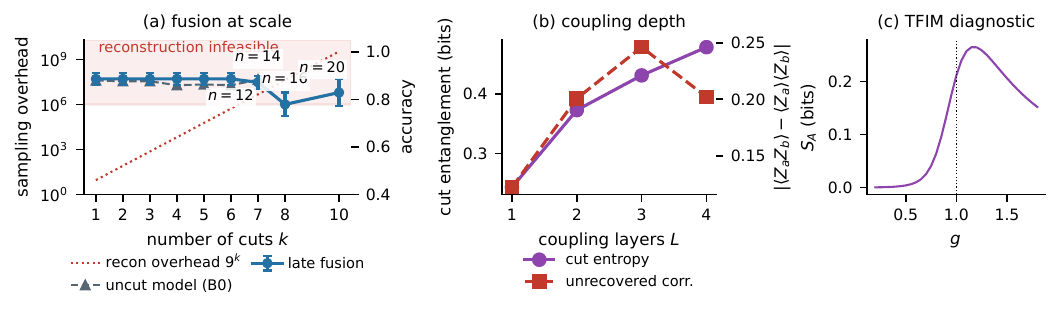}
\caption{Scaling. (a) Late fusion (blue) keeps training as cuts and width grow to $k{=}10$,
$n{=}20$, matching the uncut model (gray) wherever that reference is still trainable (the
$n{=}14$, $k{=}7$ reference is the three-seed mean under a $150$-iteration budget),
while
reconstruction's $9^k$ overhead (dotted) enters the infeasible regime (shaded). (b) With more
coupling layers, cut entanglement (purple) grows monotonically and the connected correlation
fusion cannot recover (red) grows then saturates in the $0.20$--$0.25$ band from $L{=}3$
(untrained circuits; the trained-model study in the text confirms the trend and shows that
uncoupled fusion accuracy does not degrade along it).
(c) The TFIM diagnostic $S_A(g)$ peaks near the critical point $g{=}1$ (finite-size shift).}
\label{fig:scaling}
\end{figure*}

\subsubsection{Classical inputs are always locally sufficient: a train-free proxy and the regime map}
A diagnostic that needs no quantum training at all would be the most useful kind. We therefore
test whether a quantum-free, data-side proxy predicts the fusion gap on classical inputs,
using the classical split-feature fusion baseline, and relate the outcome to the
frozen-fusion (S1) and uncoupled-fusion (S2) results above. The natural quantum-free
diagnostic is data-side: it is the accuracy that a tuned classical model loses when it is
restricted to the same feature partition that the registers see.

\paragraph{The classical split-feature gap is zero at every \texorpdfstring{$\alpha$}{alpha}.}
To this end, we built the classical analogue of Alg.~\ref{alg:fusion}: one classical
encoder per register emits $n_j$ bounded features, and a head of the same capacity as
$f_\theta$ combines them. We compared it with a tuned classical model on the full input over
the same $104$ runs (\texttt{exp\_classical\_split.py}). The gap is zero or negative at every
$\alpha$: its mean is $-0.02\pm0.05$, and split fusion reaches $0.971$ against $0.951$ for
the full model. The gap carries no knee signal ($\rho{=}-0.06$, CI $[-0.26,+0.13]$). This
outcome is not a failure of the proxy but a property of the task family. The label depends on
$x$ only through one scalar per register, $u{=}w_A\!\cdot\!x_A$ and $v{=}w_B\!\cdot\!x_B$. A
classical late-fusion model whose encoders can pass their inputs through is therefore lossless
at every $\alpha$ by Proposition~\ref{prop:lossless}(i). In other words, classical data split
across registers is always locally sufficient.

\paragraph{The frozen-fusion gap on classical data is an expressivity gap.}
Two consequences follow. First, the frozen-fusion gap on classical data
(Fig.~\ref{fig:independent}) is an \emph{expressivity} gap of the quantum subcircuits' local
readouts; it is not an information loss caused by skipping reconstruction. This explains why
uncoupled training, which re-encodes the same locally available information, closes most of
the gap, and why the residual deficits at parity and at $\alpha{=}0.75$ are small. Second,
$S_A$ predicts the knee because it measures how much the \emph{trained coupled circuit}
routes label information through cross-cut correlation. It is thus a property of the model,
not of the data.

\paragraph{Genuine informational nonlocality requires quantum-state inputs.}
Informational nonlocality in the strict sense arises only for quantum-state inputs, as in the
entangled boundary task above and in TFIM bond detection; there no classical proxy exists,
and only reconstruction recovers the label. The regime map therefore reads as follows. For
classical inputs, fusion with uncoupled training is the default, and the dial audits an
existing coupled model. For quantum inputs, the boundary experiment and the diagnostic apply.
The train-free proxy thus confirms that classical inputs are always locally sufficient, so
the residual fusion gaps on classical data are expressivity effects rather than information
losses.

\paragraph{Summary.}
Overall, our findings suggest that information is lost only when the label resides in a
cross-cut correlation of a quantum-state input, which is exactly the case that
Proposition~\ref{prop:lossless} singles out. On classical inputs the registers always see
locally sufficient features, so what fusion loses there is expressivity of the local
readouts, which uncoupled training largely restores.

\subsection{Noise and hardware}
Having located the information-theoretic boundary, we next examine how the two
readouts behave under the statistical and device noise of real execution, first in simulation
and then on live IBM Heron devices.

\subsubsection{Fusion is more robust to finite-shot noise}
Finite shots and device noise affect the two readouts differently, because only
reconstruction forms a signed sum. The classifier-level comparison below uses uncoupled fusion
(S2); the device-noise comparison is at the level of a single expectation value.

\paragraph{Fusion is more accurate than reconstruction at realistic shot budgets.}
Reconstruction's signed sum amplifies shot noise by $\gamma^2$. Figure~\ref{fig:noise}(a)
depicts accuracy against the shot budget. We can observe that, at shot budgets typical of
current hardware ($\le\!1000$), fusion is \emph{more accurate} than reconstruction. At $32$
shots, for example, fusion reaches $0.85$ against $0.66$ for reconstruction, a paired gap of
$+0.19$ with $p{=}0.002$. The two readouts converge only in the many-shots limit. Note that
per-expectation shot counting favors reconstruction, since it executes $16$ subexperiments
against fusion's $2$ settings. At equal \emph{total} budget the disparity is more pronounced:
fusion reaches $0.85$ against $0.56$ at $64$ executions, and reconstruction needs
${\sim}6{\times}10^{4}$ executions to close the gap.

\paragraph{Reconstruction carries \texorpdfstring{$1.6\times$}{1.6x} to \texorpdfstring{$6.0\times$}{6.0x} more error under device noise.}
The same fragility holds under \emph{device} noise. This comparison is at the level of the
\emph{estimation error of a single expectation value}, that is, a reconstructed global
observable against one local fusion feature; it does not concern downstream classifier
accuracy. We tested two Aer noise models: a depolarizing-plus-readout model and a \emph{real}
IBM fake-backend calibration (\texttt{FakeManilaV2}), with circuits transpiled to its native
basis. Both were run against a \emph{depth-matched} fusion probe with the gate count of the
actual fusion subcircuit. Figure~\ref{fig:noise}(b) shows that under both models the
reconstructed estimate carries more error, and that the excess is roughly flat in $k$. The
ratio is $1.6\times$ under the depolarizing model ($0.072$ against $0.046$) and $6.0\times$
under the real device calibration ($0.059$ against $0.010$). The excess arises because the
$\gamma$-weighted signed sum amplifies the fixed per-subexperiment bias, whereas a fusion
feature is one subcircuit measurement, so its error is independent of the cut count. These
results show that at realistic shot budgets fusion is more accurate than reconstruction, and
that it carries less estimation error under both device-noise models.

\begin{figure*}[t]
\centering
\includegraphics[width=\panelfigwidth]{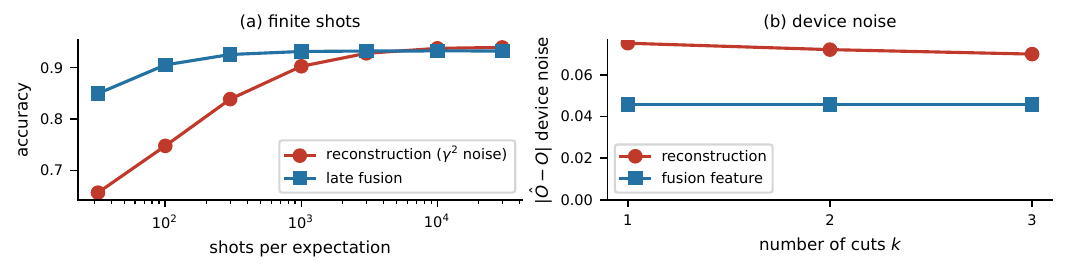}
\caption{Noise robustness. (a) Under finite shots, late fusion (blue) is more accurate than
reconstruction (red), which pays a $\gamma^2$ variance penalty. (b) Under device noise,
the estimation error of a reconstructed expectation value stays above that of a depth-matched
fusion feature ($1.6\times$ here; $6\times$ under a real IBM calibration), which is
independent of the cut count; this panel compares feature error, not classifier accuracy.}
\label{fig:noise}
\end{figure*}

\subsubsection{Validation on live hardware}
The last rung of the noise ladder is a live device, where we again compare the estimation
error of one reconstructed expectation value with that of one fusion feature. We used a live
IBM Heron device (\texttt{ibm\_fez}), with each half-width subcircuit transpiled and run
independently.

\paragraph{At equal budget the reconstruction error grows exponentially on the device.}
With \emph{full} per-subexperiment shots, reconstruction and a fusion feature sit at
comparable error over five instances; at $k{=}2$, reconstruction reaches $0.017$ against
$0.020$ for fusion, at $36\times$ the number of executions. At an \emph{equal} $8000$-shot
budget over three instances, the exponential growth of the reconstruction error is present on
the device. The reconstruction error grows from $0.023$ to $0.091$ to $0.223$ for
$k{=}1,2,3$, while the fusion error stays within $0.013$--$0.034$; the gap is already
$7\times$ at $k{=}2$.

\paragraph{The direction of the gap is robust, but its multiplier is not.}
Table~\ref{tab:hw} lists the per-instance spread, which is wide. At equal budget, the
reconstruction-to-fusion error ratio ranges over $3.5$--$24.5\times$ at $k{=}2$ and over
$3.9$--$10.7\times$ at $k{=}3$ across instances. At $k{=}1$ it straddles $1$, ranging over
$0.6$--$1.7\times$. The \emph{direction} of the gap at $k\ge2$ is therefore robust across
instances, but the headline multiplier is not a stable device constant. Note that the first
campaign used one device (\texttt{ibm\_fez}), one calibration window, and one transpilation
setting (optimization level $3$, default seed).

\paragraph{The equal-budget error gap reproduces on three devices.}
To measure that sensitivity directly, we repeated the equal-budget protocol on three
backends (Table~\ref{tab:hwmulti}): a second calibration day of \texttt{ibm\_fez}, a second
Heron r2 device, \texttt{ibm\_marrakesh}, and a Heron r3 device, \texttt{ibm\_pittsburgh}.
Each backend was run with two transpiler seeds and three instances per seed, and every run
logged the backend calibration timestamp, the transpiler seed, and the job identifiers. We can
see from the table that the direction of the gap is stable. Across all $21$ rows at $k{=}3$
the reconstruction error exceeds the fusion error, with ratios between $3.9\times$ and
$38.6\times$, and at $k{=}2$ it does so in $20$ of $21$ rows. The magnitude is not stable: at $k{=}3$ the mean
ratio is $6.7\times$ and $11.1\times$ on the two \texttt{ibm\_fez} seeds of the second day,
$20.3\times$ and $21.3\times$ on \texttt{ibm\_marrakesh}, and $14.6\times$ and $16.9\times$
on \texttt{ibm\_pittsburgh}. The spread is driven mainly by how accurately the fusion feature
itself is estimated ($0.015$--$0.031$ across devices) rather than by the reconstruction
error, which is $0.16$--$0.31$ everywhere. At $k{=}1$ the ratio straddles $1$ on every
device. The equal-budget error gap at two and three cuts is thus reproduced in its direction
on every device, day and compilation we tried. Its magnitude varies by an order of magnitude
and should be read as a range rather than as a constant.

\begin{table*}[t]
\caption{Equal-budget ($8000$ shots per readout) reconstruction-to-fusion error ratio across
devices (\texttt{ibm\_fez}, \texttt{ibm\_marrakesh}: Heron r2; \texttt{ibm\_pittsburgh}: Heron
r3), calibration days (d1, d2) and transpiler seeds (s0, s1), at optimization level $3$ with
three random instances per row. Each cell gives the mean and the per-instance range
$[\min,\max]$ of the ratio of the reconstruction error to the fusion-feature error, where both
errors are $|\hat{\langle O\rangle}-\langle O\rangle_{\mathrm{exact}}|$ of one expectation
value.}
\label{tab:hwmulti}
\centering\scriptsize
\setlength{\tabcolsep}{3pt}
\begin{tabular}{@{}lccc@{}}
\toprule
backend, day, seed & $k{=}1$ & $k{=}2$ & $k{=}3$ \\
\midrule
\texttt{fez}, d1, s0 (Table~\ref{tab:hw}) & $1.2\;[0.6,1.7]$ & $11.9\;[3.5,24.5]$ & $7.2\;[3.9,10.7]$ \\
\texttt{fez}, d2, s0 & $0.8\;[0.5,1.4]$ & $5.3\;[0.8,12.4]$ & $6.7\;[5.3,8.3]$ \\
\texttt{fez}, d2, s1 & $1.8\;[1.0,3.2]$ & $5.4\;[3.5,8.7]$ & $11.1\;[6.5,19.6]$ \\
\texttt{marrakesh}, d2, s0 & $2.9\;[2.3,3.4]$ & $7.4\;[4.5,12.8]$ & $20.3\;[8.1,36.5]$ \\
\texttt{marrakesh}, d2, s1 & $2.0\;[0.9,4.1]$ & $4.8\;[2.4,9.0]$ & $21.3\;[11.2,38.6]$ \\
\texttt{pittsburgh} (r3), d2, s0 & $1.7\;[1.0,2.2]$ & $7.7\;[4.3,12.6]$ & $14.6\;[7.4,27.4]$ \\
\texttt{pittsburgh} (r3), d2, s1 & $1.3\;[1.0,1.8]$ & $8.3\;[2.9,16.9]$ & $16.9\;[8.8,32.6]$ \\
\bottomrule
\end{tabular}
\end{table*}

\paragraph{The intermediate dial point executes on hardware with no accuracy penalty.}
Lastly, we executed the reconstruction dial itself on \texttt{ibm\_fez} for one trained
coupled model ($\alpha{=}0.5$, $k{=}2$, $36$ QPD terms) on $40$ held-out inputs with $1000$
shots per circuit. For each budget only the retained subexperiments were run, and the head
fitted on exact training features was applied to the device-estimated features. At $Q{=}0$
(two local circuits per input) the device features gave $0.925$ test accuracy against
$0.950$ with exact features. At $Q{=}0.3$ (three of $36$ terms retained, overhead $1.8$,
$240$ circuits in total) the device features gave $0.950$, equal to the exact value. At
$Q{=}1$ (all $36$ terms, overhead $14.5$, $2880$ circuits) the device features gave $0.950$
against $0.925$ exact. The intermediate point of the dial therefore executes on hardware at
its stated circuit count with no accuracy penalty on this task, which answers directly how the
dial behaves when implemented for $0<Q<1$.

\begin{table*}[t]
\caption{Live-hardware per-instance spread (\texttt{ibm\_fez}). ``Full shots'': $4000$
shots per subexperiment ($12$, $72$ subexperiments for $k{=}1,2$) and per fusion circuit.
``Equal $8000$'': the same total budget split over all subexperiments ($666$, $111$, $18$
shots each for $k{=}1,2,3$) vs.\ $4000$ per fusion circuit. Error is $|\hat{\langle
O\rangle}-\langle O\rangle_{\mathrm{exact}}|$ of one expectation value, not accuracy.}
\label{tab:hw}
\centering\scriptsize
\setlength{\tabcolsep}{3pt}
\begin{tabular}{@{}lccccc@{}}
\toprule
budget & $k$ & inst. & rec.\ error [min,max] & fusion error [min,max] & ratio [min,max] \\
\midrule
full shots & 1 & 5 & $.029$ $[.022,.044]$ & $.017$ $[.004,.036]$ & $2.8$ $[0.6,5.8]$ \\
full shots & 2 & 5 & $.017$ $[.008,.027]$ & $.020$ $[.015,.033]$ & $0.9$ $[0.5,1.7]$ \\
equal $8000$ & 1 & 3 & $.023$ $[.013,.032]$ & $.021$ $[.014,.028]$ & $1.2$ $[0.6,1.7]$ \\
equal $8000$ & 2 & 3 & $.091$ $[.077,.100]$ & $.013$ $[.004,.022]$ & $11.9$ $[3.5,24.5]$ \\
equal $8000$ & 3 & 3 & $.223$ $[.159,.271]$ & $.034$ $[.022,.041]$ & $7.2$ $[3.9,10.7]$ \\
\bottomrule
\end{tabular}
\end{table*}

\paragraph{Summary.}
Overall, our findings suggest that the signed quasiprobability sum that makes reconstruction
exact in the noiseless limit is also what makes it fragile. It amplifies shot noise by
$\gamma^2$ and amplifies the per-subexperiment bias of a device, so at realistic budgets
fusion is both more accurate and more precisely estimated. The direction of the equal-budget
error gap at two and three cuts reproduced on every device, calibration day and compilation
we tried, while its magnitude should be read as a range rather than as a constant.

\subsection{Scaling and depth}
The noise results concern fixed, small circuits. Here we examine how the two
readouts behave as the width, the number of cuts, and the coupling depth grow, first at widths
where an exact reference exists and then beyond them.

\subsubsection{Fusion where reconstruction is infeasible}
Beyond the widths at which an exact reference exists, the relevant question is whether
uncoupled fusion (S2) remains trainable at all.

\paragraph{Each half-width subcircuit sits at exponentially larger gradient variance.}
The gradient variance of a random ansatz at $n=2$--$10$ qubits decays exponentially in the
qubit count, at a slope of $-0.47$ per qubit; this is a barren plateau~\citep{mcclean2018}.
Each half-width subcircuit therefore operates at exponentially larger gradient variance than
the full circuit. Note that this is a corollary of cutting, inferred from the variance scaling
rather than from a comparative training study. Because fusion never builds the joint state,
it keeps training at scales where reconstruction cannot follow.

\paragraph{Uncoupled fusion keeps training up to \texorpdfstring{$n{=}20$}{n=20} and \texorpdfstring{$k{=}10$}{k=10}.}
Figure~\ref{fig:scaling}(a) depicts the scale-up. \CUTscaleup{} To extend the exact
reference one step further, we trained the uncut model at $n{=}14$ ($7{+}7$ qubits) for
$k\in\{1,4,7\}$ with three seeds, first under the same $35$-iteration budget and then under
a $150$-iteration budget. Under the short budget the reference reaches $0.822$--$0.861$ and
is clearly under-trained. Under the long budget it reaches $0.839{\pm}0.021$,
$0.828{\pm}0.021$ and $0.883{\pm}0.036$ for $k{=}1,4,7$, with final training losses of
$0.09$--$0.17$, at a cost of $215$--$244$ minutes per run. Uncoupled fusion reaches
$0.872{\pm}0.028$ at every $k$ in $20$--$24$ minutes. The pooled difference of $+0.022$
favours fusion, and the only cell in which the reference leads, $k{=}7$ by $0.011$, lies
within the three-seed spread. At $n{=}14$ the reference therefore brackets fusion rather than
bounding it from above, and it does so at ten times the training time. At $n{=}16$, our
\emph{uncut statevector reference} exceeded $35$ CPU-hours without converging, while fusion
trained in minutes.

\paragraph{The reference-free cells are limited by sample size, not by the method.}
Two anchors place these reference-free cells. The generator's empirical ceiling is
${\approx}0.98$, estimated from $10^4$ samples, and a classical anchor at the experiment's
$N{=}160$ reaches $0.86{\pm}0.02$. Fusion's $0.78$--$0.83$ is therefore limited by the sample
size rather than by the method. At smaller widths ($n{=}4,6,8$, single cut, $5$ seeds), the
fusion--reconstruction gap does not grow: it is $0.00$, $0.02$, and $0.01$, respectively.

\paragraph{Cut entanglement grows with depth while the unrecoverable correlation saturates.}
We also probed deeper circuits by sweeping $L=1$--$4$ coupling layers.
Figure~\ref{fig:scaling}(b) shows that, on untrained random circuits, cut entanglement grows
monotonically with $L$, from $0.24$ to $0.48$ bits, and that the connected cross-cut
correlation that a marginal-feature head cannot recover grows to ${\sim}0.25$ by $L{=}3$ and
then saturates. Reconstruction stays exact at any depth, at a cost of $9^{kL}$. Uncoupled
fusion therefore remains trainable up to $n{=}20$ and $k{=}10$, where the uncut reference is
no longer computable.

\paragraph{Deeper trained coupling builds entanglement but costs uncoupled fusion nothing.}
To test whether this structural trend translates into accuracy loss, we \emph{trained}
coupled QNNs with $L{=}1$--$4$ coupling layers. The grid covered $k\in\{1,2\}$, $n{=}4$, and
$\alpha\in\{0.25,\dots,1\}$, with $10$ seeds, for $320$ runs in total
(\texttt{exp\_deep\_coupling.py}). We compared these models with uncoupled fusion (S2) of
matched depth and with frozen fusion (S1) on the subcircuits of the trained deep circuit.
Figure~\ref{fig:deep} depicts the outcome. The trained circuits do build more cut
entanglement with depth: $S_A$ rises along $0.06\to0.63$ bits at $k{=}1$ and along
$0.08\to0.86$ bits at $k{=}2$, with Spearman $\rho(L,S_A){=}0.71$, and the unrecoverable
connected correlation rises from $0.02$ to $0.18$. Depth is therefore the same quantumness
axis as data entanglement. \emph{Uncoupled} fusion, however, loses nothing along this axis.
We can see that its gap to the coupled reference is within $0.01$ of zero at every $L$
($\rho(L,\mathrm{gap}){=}0.03$, $p{=}0.57$), and that $S_A$ does not predict it
($\rho{=}-0.02$). Only \emph{frozen} fusion degrades with depth: its gap rises from $0.05$
to $0.17$ for the cut-free construction used elsewhere in this paper
($\rho(S_A,\mathrm{gap}){=}0.33$, $p{<}10^{-8}$). The coupled models themselves train without
difficulty at these widths and plateau at ${\approx}0.93$ regardless of $L$. These results
suggest that the extra entanglement that deeper coupling generates is superfluous to the
task, which is precisely why uncoupled fusion, which never builds it, pays no price. Accuracy
loss is governed by the cross-cut information that the \emph{task} requires, namely $\alpha$
and the entangled-data boundary, and not by how much entanglement a coupled architecture
happens to generate.

\begin{figure*}[t]
\centering
\includegraphics[width=\widefigwidth]{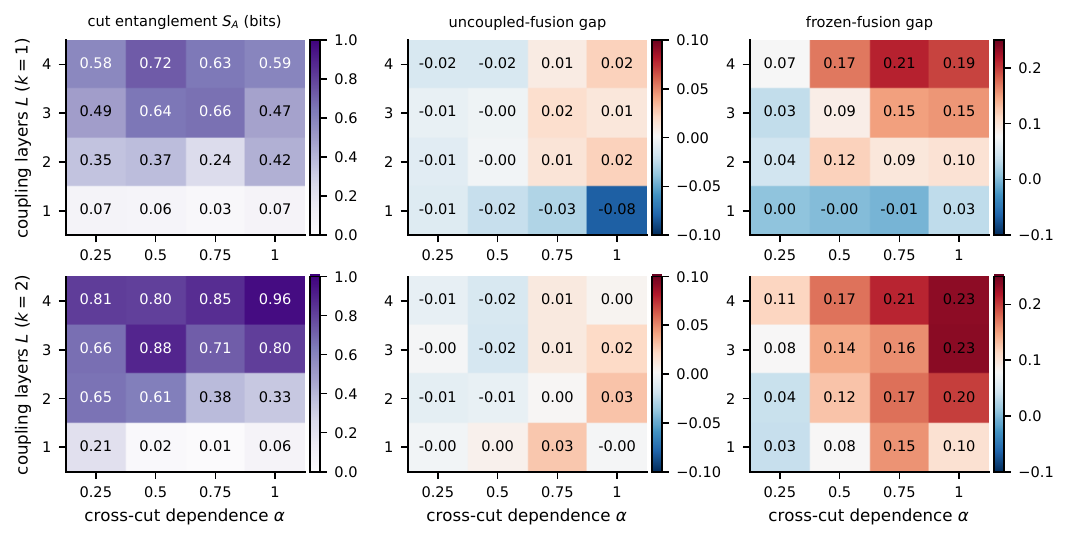}
\caption{Trained deep-coupling study ($n{=}4$, $10$ seeds per cell): trained cut entanglement
$S_A$ rises with the number of coupling layers $L$ (left), the uncoupled-fusion gap to the
coupled reference stays within $\pm0.03$ everywhere (middle), and only frozen fusion on the
deep circuit's subcircuits degrades (right). Rows: $k{=}1$ (top), $k{=}2$ (bottom).}
\label{fig:deep}
\end{figure*}

\paragraph{Summary.}
Overall, our findings suggest that late fusion is the readout that survives scale: it keeps
training to $n{=}20$ and $k{=}10$, where the uncut reference is no longer computable and the
$9^k$ overhead of reconstruction is far out of reach. Depth is the same quantumness axis as
data entanglement, but the entanglement that deeper coupling generates is superfluous to the
task, which is why uncoupled fusion, which never builds it, pays no price along it.

\subsection{The dial and the diagnostic}
The results so far compare the two endpoints of the readout family. The remaining question is what lies between them. This subsection examines the
reconstruction dial that interpolates between the endpoints, first in the physical
quasiprobability basis and then on a dense diagnostic grid. It then tests whether the
cut-entanglement diagnostic predicts where the interpolation saturates.

\subsubsection{The dial in the physical basis}
We examine what the reconstruction dial measures when it is executed in the physical
quasiprobability basis. The dial (S3) operates on the trained coupled circuits, with frozen
fusion (S1) at $Q{=}0$ and uncoupled fusion (S2) as the deployable comparator. Three
observations establish what the dial is and what it is not.

\paragraph{The dial is a physical budget.}
Figure~\ref{fig:qdialphys} depicts the dial executed through \texttt{qiskit-addon-cutting}
on the trained coupled circuits of the $\alpha$ sweep at $k{=}2$, which yields $36$ QPD
terms. We ran $5$ seeds with the exact estimator and $3$ seeds at $4000$ shots per
subexperiment. The partial sums reproduce the exact expectations to $2{\times}10^{-15}$, and
they reproduce the library's own reconstruction identically. Every retained term is a
subexperiment that the device runs, and the cost axis is the Monte-Carlo overhead
$(\sum_{\text{kept}}|a_i|)^2$. At the \emph{trained} coupling angles the full overhead is
$\kappa^{2k}=35$--$81$; the worst case $9^k{=}81$ is reached only at $\phi{=}\pi/2$. This is
the same order as the Schmidt convention's $39$--$64$.

\paragraph{The dial and uncoupled fusion are two different objects.}
The endpoints make this clear. The fixed-head reconstruction (B1) and the retrained head at
$Q{=}1$ coincide exactly in $9$ of $25$ runs and differ by at most $0.02$ on average, with the
largest differences at $\alpha\ge0.5$; the augmented readout family therefore changes little
at the top. The $Q{=}0$ point is frozen fusion, at $0.74$--$0.93$. \emph{Uncoupled} fusion,
by contrast, sits at $0.935$--$0.960$ at zero reconstruction cost; this is within $0.005$ of
the $Q{=}1$ accuracy for $\alpha\le0.75$ and $0.02$ below it at parity. In other words, the
dial audits how much of a \emph{given coupled model's} reconstruction budget carries label
information, whereas uncoupled training is the deployable alternative that needs none of that
budget and reaches the top of the frontier on its own.

\paragraph{The knee is robust to the convention.}
The sufficient budget $Q^\star$ is $0.10$--$0.12$ for $\alpha\le0.5$ and $0.28$--$0.32$ for
$\alpha\ge0.75$ in the physical basis, against $0.10$--$0.26$ in the Schmidt basis; out of
$25$ runs, $7$ are higher and $4$ are lower in the physical basis. In physical cost, the knee
uses $1$--$4$ units of overhead at $\alpha\le0.5$ and $18$--$19$ units at $\alpha\ge0.75$,
that is, between $3\%$ and about half of the full $\kappa^{2k}$. At $4000$ shots the executed
curve tracks the exact one within ${\sim}0.02$ for $Q\ge0.2$. Below that budget, shot noise
on the few retained signed terms costs up to $0.13$ at $\alpha{=}0.75$. The dial is thus a
physical budget whose knee is stable across conventions, and uncoupled fusion reaches the top
of its frontier without any reconstruction budget.

\begin{figure}[t]
\centering
\includegraphics[width=\singlefigwidth]{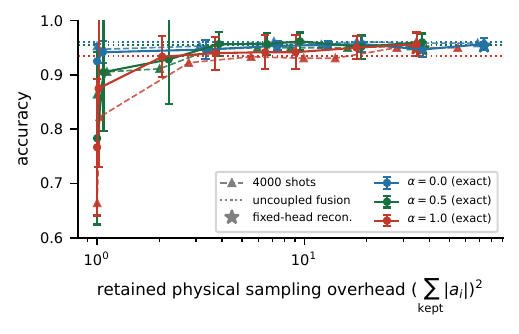}
\caption{The dial in the physical QPD basis (\texttt{qiskit-addon-cutting}, $k{=}2$):
accuracy versus retained Monte-Carlo overhead $(\sum_{\text{kept}}|a_i|)^2$ for the exact
estimator (solid, $5$ seeds) and $4000$ shots per subexperiment (dashed, $3$ seeds). Dotted:
uncoupled fusion at zero reconstruction cost; stars: fixed-head reconstruction (B1).}
\label{fig:qdialphys}
\end{figure}

\subsubsection{The diagnostic grid}
We test on a dense grid whether the cut-entanglement diagnostic $S_A$ of the trained coupled
circuit predicts where the dial saturates, using the dial (S3) in the Schmidt convention with
frozen fusion (S1) at $Q{=}0$.

\paragraph{The knee rises with cross-cut dependence.}
Figure~\ref{fig:qdial} depicts the frontier on the $104$-run grid, where the dial traces the
same frontier in the Schmidt convention. Sweeping $Q$ from $0$ (fusion only) to $1$ (full
reconstruction) recovers accuracy while the retained mass climbs from a constant to
$\gamma_S^{2k}$. Note that the $Q{=}0$ endpoint is \emph{frozen} fusion, not uncoupled fusion
(cf.\ Fig.~\ref{fig:independent}); this is the origin of the small nonzero knee at low
$\alpha$. We define the \emph{knee} as the smallest $Q$ within $0.02$ of the $Q{=}1$
accuracy. We can observe that the knee rises with the cross-cut dependence $\alpha$, from
$0.08$--$0.11$ at $\alpha\le0.25$ to $0.16$--$0.20$ at $\alpha\ge0.5$.

\paragraph{Trained cut entanglement predicts the knee at \texorpdfstring{$\rho{=}0.59$}{rho=0.59}.}
On a dense grid of $13$ $\alpha$-values $\times$ $8$ seeds, that is, $104$ runs, all at
$n{=}4$, the trained circuit's cut entanglement $S_A$ correlates with the knee at Spearman
$\rho{=}0.59$ ($95\%$ CI $[0.46,0.68]$) and with the fusion gap at $\rho{=}0.44$ (CI
$[0.26,0.59]$). It is thus a moderately strong knee predictor and a moderate gap indicator.
Partialing out $\alpha$ leaves $\rho{=}0.60$ $[0.43,0.69]$, positive in every $\alpha$ cell,
so the correlation is not the shared $\alpha$-trend. Within this task family and at fixed
width, then, the cut-entanglement diagnostic predicts the position of the knee.

\begin{figure}[t]
\centering
\includegraphics[width=\singlefigwidth]{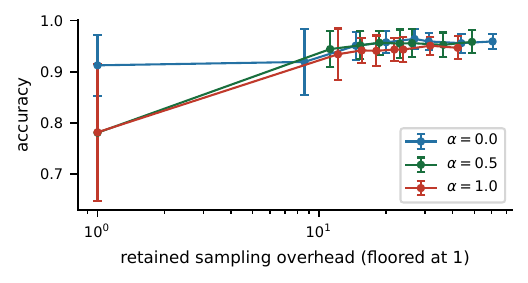}
\caption{$Q$-dial accuracy/cost frontier in the exact-simulation (Schmidt-mass) convention
(mean$\pm$std, $10$ seeds): accuracy saturates well before full reconstruction, so most of the
$\gamma_S^{2k}$ overhead yields little additional accuracy; low-cost fluctuations are within
seed noise. Fig.~\ref{fig:qdialphys} shows the physically executed counterpart.}
\label{fig:qdial}
\end{figure}

\paragraph{The diagnostic is a within-scale, post-hoc indicator only.}
Three negative results bound the use of the diagnostic. (a) It requires the \emph{trained
coupled} circuit. $S_A$ at initialization, or after only five optimizer iterations, carries no
knee signal (Spearman $|\rho|{<}0.1$ on the same $104$ runs). (b) It did not transfer across
scale in our one test. $S_A$ of the trained coupled circuit at $n{=}8$ with $2$ cuts does not
rank-order the uncoupled-fusion gap of the same task family at $n{=}12$ ($\rho{=}-0.09$, not
significant, $n{=}12$ runs). This is partly because that gap is already small and nearly flat
at $n{=}12$, with means in $[-0.006,+0.022]$. (c) Its validation rests mainly on a single synthetic
task family at a single width. A cross-family check on the four standard datasets of
Table~\ref{tab:bench} ($40$ runs) gives a weaker but significant knee correlation,
$\rho{=}0.37$ ($p{=}0.02$), and no correlation with the fusion gap ($\rho{=}-0.13$, not
significant). We therefore report $S_A$ as a \emph{within-scale, post-hoc}
indicator for a trained coupled circuit, not as a task-general predictor, and we do not
describe the method as self-characterizing. In the regime that motivates fusion
($n{=}16$--$20$, $k{=}8$--$10$; Fig.~\ref{fig:scaling}(a)), the coupled circuit is not
trainable at all, so $S_A$ is unavailable there. The Limitations and Future Work section lists the routes that remain:
pilot-scale measurement, randomized-measurement entropy estimation on
hardware~\citep{brydges2019,elben2020}, and quantum-free data-side proxies, which were tested
in the regime-map subsection above.

\paragraph{Summary.}
Overall, our findings suggest that the dial is a physical budget whose knee is stable across
conventions and rises with cross-cut dependence, and that most of the reconstruction overhead
yields little additional accuracy. The trained circuit's cut entanglement predicts where the
knee sits within a task family at fixed width, but it requires the trained coupled circuit,
did not transfer across scale in our one test, and is unavailable in the regime that
motivates fusion.

\subsection{Prior-art readouts}
Lastly, we compare the reconstruction dial and the trained fusion head with the
closest prior-art readouts, namely a learned truncation of the reconstruction series and the
parameter-free Kawase sum.

\subsubsection{Against learned truncation}
The closest existing alternative to the dial is a learned truncation of the reconstruction
series, and we compare the two at matched budget.

\paragraph{The dial leads learned truncation at every matched budget.}
A Marshall-style readout~\citep{marshall2023} places \emph{learned} weights on the top-$T$
reconstruction terms and is the dial's closest alternative. Pooled over $\alpha$ at matched
budgets, the dial leads at every $T$ over $10$ seeds; at $T{=}1$ the paired gap is $+0.04$
with $p{<}10^{-3}$. The largest margins occur at small budgets on hard tasks, where the dial
reaches $0.74$ against $0.63$ at $T{=}1$ and $\alpha{=}0.75$. The advantage is not merely a
matter of extra inputs: granting the truncation the same raw fusion features recovers little,
at $0.64$ on that cell.

\paragraph{The trained head leads the fixed sum wherever interactions matter.}
The parameter-free Kawase sum (B3) and the entanglement ablation (B4) are compared in
Tables~\ref{tab:bench} and~\ref{tab:faithful}. The trained head's margin over the fixed sum
is largest where feature interactions matter, at $+0.26$ on Breast-Cancer and $+0.34$ on the
synthetic task with $\alpha{=}0.75$, and it vanishes on the local image tasks.

\paragraph{Summary.}
Overall, our findings suggest that the dial leads learned truncation at every matched budget,
and that the trained fusion head leads the parameter-free sum wherever cross-cut interactions
matter. We hypothesize that both advantages share one origin: a trained head can weight and
combine local features nonlinearly, whereas a fixed sum or a truncated signed series cannot.

\section{Discussion}
This section places our results against classical baselines and against the question of
classical simulability.

\paragraph{The tension between local information and classical sufficiency.}
Our results must be read against classical baselines. Tuned classical baselines match or beat
the quantum model on every classical dataset~\citep{bowles2024}; the advantage of fusion is
therefore over reconstruction, not over classical machine learning. In turn, a
fusion--reconstruction tie \emph{indicates} that the task's label information is recoverable
from local marginals, so that the pipeline gains nothing from its cross-cut correlation. The
comparison thereby doubles as an audit instrument. Fusion's regime, in which the task
information is local, overlaps the regime where classical models succeed, whereas a
quantum--classical separation would require cross-cut information, which is exactly where
fusion fails. We therefore claim \emph{efficiency and diagnosis within QML}, not an advantage
over classical learning.

\paragraph{Local readouts and classical simulability.}
A related question, raised by \citet{cerezo2025simulability}, is whether a model that avoids
barren plateaus by remaining local is also classically simulable. Our fusion model is local
by construction: each subcircuit is a small circuit whose readout is a bounded function of its
own inputs. The classical split-feature baseline in the Results shows that, for classical
inputs, a classical encoder of the same width loses nothing. We therefore do not claim that
uncoupled fusion is beyond classical reach on classical data; it is not, which is consistent
with the absence of any advantage over tuned classical models in Tables~\ref{tab:bench}
and~\ref{tab:faithful}. The value of fusion lies elsewhere. Given a circuit-cutting pipeline,
it removes the exponential reconstruction cost without changing what the pipeline can learn.
For quantum-state inputs, where the question is which quantum readout to pay for rather than
whether to use one, the boundary experiment shows that the choice between local readout and
reconstruction has a non-trivial answer that the locality condition predicts.

\section{Limitations and Future Work}
This section discusses the limitations of our study and points to several directions for
future work. Four limitations bound the present study; we describe each below, together
with the direction for future work that it suggests.

\paragraph{Scale.}
Fusion trains to $20$ qubits and $10$ cuts, and the $9^k$ overhead scaling is verified to
$k{=}5$. The last fully converged exact \emph{accuracy} reference is at $12$ qubits and $6$ cuts; at
$14$ qubits the reference trains under a $150$-iteration budget to within the seed spread of
fusion, at ten times the training time, and from $16$ qubits it is unavailable.

\paragraph{Coupling depth.}
Trained deep-coupling models were studied only at $n{=}4$ and $L\le4$, where the task
saturates near $0.93$. A width at which deeper coupling is \emph{needed} for accuracy would
provide a sharper test, and finding one is an interesting direction for future work.

\paragraph{Noise is mostly simulated.}
The noise studies use finite shots together with the two Aer device-noise models of
Fig.~\ref{fig:noise}(b). Live-device checks comprise $10$ full-shot runs with $k\le2$ and $9$
equal-budget runs with $k\le3$ on one device in the first campaign. A second campaign added
$54$ equal-budget runs across three IBM Heron backends (two Heron r2 devices and one Heron r3
device), two calibration days and two transpiler seeds, and one execution of the
reconstruction dial on a device. These runs establish the direction of the hardware gap at
$k\ge2$ but show that its magnitude varies by an order of magnitude across instances, devices
and compilations. Device biases partially cancel in the QPD sum; see the appendix. Broader
hardware campaigns that pin down the magnitude of this gap are an interesting direction for
future work.

\paragraph{Diagnostic cost and theory.}
The cut-entanglement diagnostic presupposes a trained \emph{coupled} circuit. At the scales
that motivate fusion, this is exactly the circuit that cannot be trained. Its validation is
within one scale and one task family, and the one cross-scale test was null (Results). Three
routes remain open. The first measures $S_A$ at a pilot width where the coupled circuit is
trainable, although transfer is unproven. The second estimates cut entropy on hardware from
randomized measurements of the coupled circuit~\citep{brydges2019,elben2020}; this verifies
the diagnostic but does not avoid the coupled run. The third uses quantum-free data-side
proxies of cross-partition label information, which need no quantum training at all. A
conjectured $S_A$-to-gap bound (appendix) awaits proof. Pursuing these routes, and proving the
conjectured bound, is an interesting direction for future work.

\section{Conclusions}
In this paper, we evaluated whether a cut quantum neural network needs exponential-cost
quasiprobability reconstruction for machine-learning tasks, or whether late fusion, a
noise-robust readout whose cost is constant in the number of cuts, suffices. Through
head-to-head comparisons of uncoupled fusion with exact reconstruction of the same cut
circuit, a reconstruction dial, and a cut-entanglement diagnostic, we charted when each
readout is needed. Our findings suggest that late fusion tracks reconstruction whenever the
task information is local (Proposition~\ref{prop:lossless}), and that it must lose accuracy
when the label depends on a decision-relevant cross-cut correlator; the dial and the
diagnostic indicate which of these regimes holds. Code and results accompany the article.
Train-free diagnostic proxies, broader hardware runs, and tasks with a \emph{local} quantum
advantage that fusion could inherit are interesting directions for future work.

\ifqmi\else
\paragraph{Use of generative AI.}
A large language model assisted with methodology assessment, code implementation, and
manuscript writing. The authors made all design and interpretation decisions, verified every
reported number against the released results, and are responsible for the entire
content (appendix~\S E).
\fi

\paragraph{Code and data availability.}
Code, per-seed result files and the figure scripts are available at
\url{https://github.com/p299sing/late-fusion-qml}. No new datasets were collected; all real
datasets are public (Iris and Breast-Cancer via scikit-learn, Moons and Circles from the
scikit-learn generators, MNIST and Fashion-MNIST from OpenML), and the synthetic tasks are
generated by the code with the seeds stated in the appendix.

\paragraph{Acknowledgements.}
The authors acknowledge the support of the IBM Quantum Hub at the University of Melbourne
for access to IBM Quantum devices.

\begingroup\small
\bibliographystyle{plainnat}
\bibliography{references}
\endgroup

\clearpage
\appendix
\section{Formal statement and proof of the locality condition}

\paragraph{Setting.}
We consider an input $x\!\sim\!\mathcal{D}$ that prepares a state $\rho_{AB}(x)$ on registers
$A,B$, with label $y\in\{0,1\}$. We fix finite local observable sets $\{O_A^i\}_{i\le d_A}$ and
$\{O_B^j\}_{j\le d_B}$ and define the \emph{local feature maps}
$m_A(x)=\big(\mathrm{tr}[\rho_A(x)\,O_A^i]\big)_i$ and
$m_B(x)=\big(\mathrm{tr}[\rho_B(x)\,O_B^j]\big)_j$, where $\rho_A=\mathrm{tr}_B\,\rho_{AB}$.
A \emph{fusion predictor} is any measurable function $f(m_A,m_B)$. A \emph{joint}
(reconstruction) predictor may additionally access the cross-cut correlators
$C(x)=\mathrm{tr}[\rho_{AB}(x)\,O_A\!\otimes\!O_B]$. We write $M=(m_A,m_B)$ and let $h_2$
denote the binary entropy.

\begin{proposition}[Losslessness]\label{prop:app-lossless}
If the Bayes-optimal predictor factors through the local features---i.e.\ there exists
measurable $g$ with $\Pr[y{=}1\mid x]=g(M(x))$ $\mathcal{D}$-a.s.---then
$\inf_f R(f(M))=R^\star$, the Bayes risk, and the infimum is approached by any universal
approximator family (e.g.\ one-hidden-layer MLPs) trained on $(M,y)$. Reconstruction cannot
improve on it.
\end{proposition}

\begin{proof}
$R^\star=\mathbb{E}\,\min(g(M),1-g(M))$ is attained by thresholding $g(M)$, a function of $M$
alone; universal approximation gives $f_\theta\to g$ in $L^1(\mathcal{D}_M)$ along a suitable
sequence, and excess risk of plug-in rules is controlled by
$\mathbb{E}|f_\theta(M)-g(M)|$. Any predictor using more information than $x$ itself cannot beat
the Bayes risk, so joint readout offers no improvement.
\end{proof}

\begin{proposition}[Converse: a residual-information lower bound]\label{prop:app-converse}
Suppose the label is determined by the correlator together with the local features,
$H(Y\mid M,C)=0$, but the correlator carries \emph{residual} label information
$\varepsilon:=I(Y;C\mid M)>0$. Then every fusion predictor $f(M)$ has error
\[
\Pr[f(M)\neq Y]\;\ge\;h_2^{-1}\!\big(H(Y\mid M)\big)\;=\;h_2^{-1}(\varepsilon),
\]
while the Bayes-with-$C$ predictor attains error $0$. Hence the excess risk of fusion is at
least $h_2^{-1}(\varepsilon)$.
\end{proposition}

\begin{proof}
By the chain rule, $H(Y\mid M)=H(Y\mid M,C)+I(Y;C\mid M)=\varepsilon$. Fano's inequality for
binary $Y$ gives, for any estimator $\hat Y=f(M)$,
$h_2(\Pr[\hat Y\neq Y])\ge H(Y\mid M)=\varepsilon$, i.e.\
$\Pr[\hat Y\neq Y]\ge h_2^{-1}(\varepsilon)$ (taking the branch $\le 1/2$). Since
$H(Y\mid M,C)=0$, the predictor thresholding $\Pr[Y{=}1\mid M,C]$ is exact.
\end{proof}

\paragraph{Exact condition for an accuracy loss (remark on Proposition~\ref{prop:app-converse}).}
Note that positive residual information $I(Y;C\mid M)>0$ does \emph{not} by itself imply a
positive excess classification risk. The correlator can shift the posterior
$\eta_{MC}=\Pr[Y{=}1\mid M,C]$ without moving it across $\tfrac12$, in which case the Bayes
decision given $M$ alone is unchanged; this happens, for instance, when the majority-class
posterior exceeds $\tfrac12$ for every value of $C$. More concretely, the exact statement is
the following. Let
$\eta_M=\Pr[Y{=}1\mid M]=\mathbb{E}[\eta_{MC}\mid M]$. The excess risk of the best fusion
predictor over the best joint predictor is
\[
\begin{split}
\Delta&=\mathbb{E}_M\big[\min(\eta_M,1-\eta_M)\big]\\
&\quad-\mathbb{E}\big[\min(\eta_{MC},1-\eta_{MC})\big]\;\ge\;0,
\end{split}
\]
by Jensen's inequality applied to the concave function $t\mapsto\min(t,1-t)$. Moreover,
$\Delta>0$ if and only if $\eta_{MC}-\tfrac12$ changes sign with $C$ on a set of positive
probability; in that case the correlator is \emph{decision-relevant}. Proposition~\ref{prop:app-converse} is the
special case $H(Y\mid M,C)=0$. In that case every $\eta_{MC}\in\{0,1\}$, decision relevance
follows from $I(Y;C\mid M)>0$, and Fano's inequality quantifies the loss. We state both the
exact condition and the deterministic special case in the main text.

\paragraph{Instantiation in the paper.}
We realize each of the two propositions by one experiment in the paper. The entangled-data
boundary experiment realizes the converse: correlators of the type $C=\langle Z_aZ_b\rangle$
carry all label information while every local marginal is uninformative. There $M$ is
independent of $Y$, so $\varepsilon=H(Y)=1$ bit, and fusion is forced to chance. The TFIM
phase task realizes the losslessness case, because the local order parameters are
sufficient. The cut-entanglement diagnostic $S_A$ enters through the overhead--entanglement
bound of Jing et al.\ (Phys.\ Rev.\ A 111, 012433): the resource that reconstruction spends
and the correlation that fusion discards are the same quantity, which motivates $S_A$ as the
diagnostic. We evaluate the diagnostic empirically, through the knee and gap correlations,
rather than claiming a tight theoretical link between $S_A$ and $\varepsilon$.

\paragraph{A conjectured gap bound.}
The converse (Proposition~\ref{prop:app-converse}) bounds the fusion excess risk by $h_2^{-1}(\varepsilon)$, where
$\varepsilon=I(Y;C\mid M)$ is the residual label information in the discarded cross-cut
correlator. This quantity is task-specific and is not directly measurable. Since $S_A$
upper-bounds the correlation available to carry that information, we conjecture a
\emph{monotone} surrogate. For the hardware-efficient ansatz with per-cut entanglement
entropies $S_A^{(j)}$, we conjecture that the fusion excess risk obeys
\[
R_{\mathrm{fusion}}-R^\star \;\le\; C\,\textstyle\sum_j\big(1-2^{-S_A^{(j)}}\big)
\]
for a task-dependent constant $C$; the summand vanishes at $S_A^{(j)}{=}0$ and saturates at
$1$. On the $104$-run dense grid we can see that this envelope is consistent with the data.
The fitted form is monotone, with Spearman $\rho{=}0.44$ between $S_A$ and the gap, and the
bound holds for $93\%$ of runs at $C{\approx}10$. The constant is loose, however, and the
scatter is large, with $R^2{=}0.10$ for a linear fit. We therefore state this bound as a
target for future theory rather than as an established result; a tight, provable version
relating $S_A$ to $\varepsilon$ remains open.

\paragraph{Cross-scale transfer of the diagnostic (empirical).}
We tested whether $S_A$ measured at a \emph{cheap} width predicts the fusion gap at a
\emph{larger} width on the same task family (\texttt{exp\_scale\_transfer.py};
$\alpha\in\{0,0.3,0.6,0.9\}$, $3$ seeds). The predictor was $S_A$ of the trained coupled
circuit at $n{=}8$ ($4{+}4$ qubits, $2$ cuts); the target was the uncoupled-fusion gap at
$n{=}12$ ($6{+}6$ qubits). At $n{=}12$ the per-$\alpha$ mean gap is small and nearly flat in
locality, with means in $[-0.006,+0.022]$ and individual runs within $\pm0.15$; uncoupled
fusion nearly matches reconstruction at this width. There is consequently little
\emph{systematic} residual for $S_A$ to track, and we observe that $S_A^{n=8}$ does not
rank-order the individual-run gaps (Spearman $-0.09$, not significant, $n{=}12$). We
therefore report $S_A$ as a \emph{within-scale} indicator; the $\rho{=}0.59$ dense-grid
correlation is at fixed $n{=}4$. We hypothesize that the shrinking gap with width reflects
the cross-cut robustness of uncoupled training improving at scale; cross-scale transfer of
the diagnostic remains an open question.

\section{Hyperparameters and experimental protocol}

\paragraph{Architecture (all statevector experiments).}
Each register carries a hardware-efficient ansatz: one $R_Y,R_Z$ pair per qubit per local
block, followed by a $CZ$ entangling chain (ansatz~A). Ansatz~B, used for the robustness
check in Section~\ref{sec:results}, replaces the rotations by one $R_X,R_Y$ pair per qubit and
the chain by a CNOT ladder, and fixes the cross gates at $\phi{=}\pi/2$. There are two local
blocks per register, one before and one after the coupling. The coupling layer consists of $k$ $R_{ZZ}(\phi)$ gates
across the boundary, with $\phi$ trainable and initialized at $\pi/2$ for the coupled model;
the coupling layer is absent for uncoupled fusion. The readout observables of the coupled
model are $\langle Z_q\rangle$ on every qubit plus the boundary $\langle Z_aZ_b\rangle$ pairs.
Fusion sees per-subcircuit $\langle Z_q\rangle$, plus $\langle X\rangle,\langle Y\rangle$
where stated.

\paragraph{Training.}
All trainers use \texttt{scipy} L-BFGS-B with finite-difference gradients. We set the
$\ell_2$ penalties to $10^{-4}$ for the fusion trainers (circuit and head) and to $10^{-3}$
for the stand-alone MLP/logistic heads and the B0 linear head. Quantum-feature caching
memoizes subcircuit features on the quantum-parameter bytes; note that this caching is exact
and has no effect on results. The fusion heads are one-hidden-layer tanh MLPs with a sigmoid
output for binary tasks or a softmax output for $K$-way tasks. The classical anchors are
SVM-RBF ($C{=}5$) and MLP $(32,16)$; we model-selected them by 4-fold CV on the training
split only and scored them on the identical held-out split.

\begin{table}[ht]
\caption{Per-experiment protocol. ``maxiter'' = L-BFGS-B iterations (coupled/fusion where two
values). $^{\ast}$At $n{=}14$: three seeds, and the uncut reference was trained at $35$ and at
$150$ iterations. For real datasets, standardization, PCA, and angle scaling are fit on the training
split only and applied to the test split (clipped to the encoding range), identically for
every model. Image datasets: PCA-4 of standardized pixels, $110$--$130$ samples/class.}
\centering\scriptsize
\setlength{\tabcolsep}{1pt}
\begin{tabular}{@{}llllll@{}}
\toprule
Experiment & data (size) & split & seeds & maxiter & head \\
\midrule
$\alpha$-sweep (H1) & synth.\ (300) & 180/120 & 10 & 90 & MLP-8 \\
benchmarks & 4 datasets & 70/30 & 10 & 90 & MLP-8 \\
independent & synth.\ (220) & 150/70 & 20 & 90/120 & MLP-8 \\
shots & Breast-C.\ (569) & 70/30 & 10 & 90 & MLP-8 \\
cost crossover & synth.\ (200) & 140/60 & 10 & 70 & MLP-8 \\
boundary & Bell states (400) & 70/30 & 10 & -- & MLP-12 \\
width $n{=}4..8$ & synth.\ (200) & 140/60 & 5 & 45/50 & MLP-(8{+}n) \\
scale-up $n{=}12..20$ & synth.\ (160--200) & 70/30 & 5$^{\ast}$ & 35/45$^{\ast}$ & MLP-(8{+}n/2) \\
faithful & F-MNIST/MNIST/synth. & 70/30 & 10 & 55/60 & MLP-8 \\
faithful PCA-8 & F-MNIST/MNIST ($4{+}4$) & 70/30 & 10 & 55/60 & MLP-8 \\
multiclass & Iris-3/MNIST-012 & 70/30 & 5 & 75 & softmax-12 \\
ablation & synth.\ (240) & 160/80 & 5 & 55 & MLP-$h$ \\
TFIM & $n{=}8$ exact diag.\ (240) & 70/30 & 5 & -- & MLP-16 \\
noise (G5) & random circuits & -- & 5 & -- & -- \\
qiskit ($k{\le}5$) & random circuits & -- & 1--3 & -- & -- \\
dense-$\alpha$ & synth.\ (300) & 180/120 & 8 & 60 & MLP-8 \\
Marshall H2H & synth.\ (300) & 180/120 & 10 & 60 & logistic \\
\bottomrule
\end{tabular}
\end{table}

\paragraph{Noise models.}
We model finite shots as Gaussian noise on expectation values, with standard deviation
$\gamma_{\mathrm{tot}}/\sqrt{s}$ for reconstruction and $1/\sqrt{s}$ for fusion, over $30$
trials per setting. The quantum-data tasks use binomial sampling with $400$ shots. We model
device noise in two ways: (i) a depolarizing model with $p_1{=}0.005$, $p_2{=}0.02$, and
readout error $0.02$; and (ii) the \texttt{FakeManilaV2} calibration, with circuits
transpiled to its native basis. Both device-noise studies use $20{,}000$ shots, $5$ seeds,
and a depth-matched fusion probe. The live-device check (\texttt{exp\_hardware.py}) ran five
random instances per $k\in\{1,2\}$ on \texttt{ibm\_fez} (IBM Heron), with $4000$ shots per
circuit. We ISA-transpiled everything at optimization level~$3$ and applied no error
mitigation; we used the raw sampler counts. We compared the QPD subexperiments ($12$/$72$
circuits) with each half-width subcircuit executed independently ($2$ circuits), both against
exact statevector references. Unlike the uniform-bias Aer models, the live per-qubit biases
vary in sign across subexperiments and partially cancel in the signed sum. At full
per-subexperiment shots, both readouts therefore sit at comparable few-percent error, with
reconstruction paying $6\times$ or $36\times$ the executions. The equal-budget variant
(\texttt{exp\_hardware\_budget.py}) split $8000$ shots per readout evenly across each
readout's circuits, that is, over $2$ fusion settings or over $2\cdot6^{k}$ subexperiments;
$k{=}3$ uses $3{+}3$-qubit registers, and we ran three instances per $k$.

\section{Ablation and multiclass details}

\paragraph{Ablation (\texorpdfstring{$5$}{5} seeds).}
On a hard cross-cut task ($\alpha{=}0.75$) we ablated the fusion feature set and the head
width. The feature sets were $\langle Z\rangle$, $\langle Z,X\rangle$, and $\langle
Z,X,Y\rangle$ per qubit, and the head widths were $h\in\{4,8,16\}$. We can observe that
accuracy is robust across the grid, ranging over $0.88$--$0.94$. Adding $\langle Y\rangle$
contributes little beyond $\langle Z,X\rangle$, whose width-$8$ cell is the best at $0.935$.
These results suggest that no single knob dominates.

\paragraph{Multiclass (\texorpdfstring{$5$}{5} seeds).}
With a $K$-way softmax fusion head over the same uncoupled subcircuit features, late fusion
reaches $0.951$ on $3$-class Iris and $0.964$ on $3$-class MNIST ($0/1/2$). These values are
within ${\sim}0.01$ of a tuned classical model, which reaches $0.964$ and $0.972$,
respectively. This shows that the method extends to multiclass tasks at the same linear
measurement cost.

\section{Paired statistics for headline comparisons}

Table~\ref{tab:paired} lists paired statistics for every headline comparison in the main
text. For each comparison it lists the mean per-seed \emph{paired} gap, a bootstrap $95\%$
CI from $10{,}000$ resamples of the paired differences, and a two-sided exact Wilcoxon
signed-rank $p$-value. We test matching claims as \emph{equivalence} at a margin fixed in
advance, $\delta{=}0.02$ accuracy, by two one-sided tests (TOST at level $0.05$). This
criterion is equivalent to the $90\%$ bootstrap CI of the paired gap lying inside
$[-\delta,\delta]$. We test superiority claims by $95\%$ CIs that exclude zero. Note that
with $n{=}10$ pairs the TOST is underpowered when the per-seed spread is wide; this is why
several Moons and single-$\alpha$ rows are inconclusive rather than equivalent. The smallest
attainable exact $p$ at $n{=}5$ pairs is $0.0625$. The tests are exploratory, and we apply no
multiple-comparison correction: with $n{\le}40$ pairs the discrete exact test cannot reach
Bonferroni-level thresholds over $34$ rows, and the bootstrap CIs carry the primary
inference. The table shows the headline subset; all $34$ comparisons regenerate from
\texttt{analyze\_significance.py}.

\ifqmi\begin{table}[ht]\else\begin{table*}[t]\fi
\caption{Paired per-seed statistics (fus.\ = fusion, rec.\ = reconstruction, class.\ = tuned
classical). Positive gaps favor the first-named method. TOST column: equivalence test at the
predefined margin $\delta{=}0.02$ (eq = equivalent, $90\%$ CI inside $[-\delta,\delta]$;
eq$^{*}$ = equivalent with a small but significant offset; inc = inconclusive at this $n$;
def/adv = significant deficit/advantage of the first-named method).}
\label{tab:paired}
\centering\scriptsize
\setlength{\tabcolsep}{1pt}
\begin{tabular}{@{}lrrlrl@{}}
\toprule
comparison (A $-$ B) & $n$ & gap & $95\%$ CI & $p$ & TOST \\
\midrule
fus.\,$-$\,rec., $\alpha$-sweep pooled & 80 & $-.010$ & $[-.019,-.001]$ & $.11$ & eq$^{*}$ \\
fus.\,$-$\,rec., $\alpha{=}1$ (parity) & 20 & $-.035$ & $[-.055,-.016]$ & $.006$ & def \\
fus.\,$-$\,rec., Moons & 10 & $+.002$ & $[-.028,+.038]$ & $.76$ & inc \\
fus.\,$-$\,rec., Circles & 10 & $-.009$ & $[-.016,-.003]$ & $.03$ & eq$^{*}$ \\
fus.\,$-$\,rec., Iris & 10 & $.000$ & $[.000,.000]$ & $1.0$ & eq \\
fus.\,$-$\,rec., Breast-C. & 10 & $-.006$ & $[-.015,+.001]$ & $.19$ & eq \\
fus.\,$-$\,rec., Fashion-MNIST & 10 & $.000$ & $[-.012,+.012]$ & $1.0$ & eq \\
fus.\,$-$\,rec., MNIST & 10 & $+.001$ & $[-.004,+.006]$ & $.81$ & eq \\
fus.\,$-$\,rec., F-MNIST (PCA-8) & 10 & $-.010$ & $[-.029,+.008]$ & $.52$ & inc \\
fus.\,$-$\,rec., MNIST (PCA-8) & 10 & $+.003$ & $[-.003,+.008]$ & $.75$ & eq \\
fus.\,$-$\,rec., synth.\ $\alpha{=}0.75$ & 10 & $-.024$ & $[-.043,-.003]$ & $.06$ & def \\
fus.\,$-$\,Kawase, Breast-C. & 10 & $+.262$ & $[+.225,+.296]$ & $.002$ & adv \\
fus.\,$-$\,Kawase, synth.\ $\alpha{=}0.75$ & 10 & $+.343$ & $[+.276,+.410]$ & $.002$ & adv \\
fus.\,$-$\,rec., $32$ shots & 10 & $+.192$ & $[+.175,+.209]$ & $.002$ & adv \\
$Q$-dial\,$-$\,Marshall, $T{=}1$ (pool.) & 40 & $+.042$ & $[+.017,+.072]$ & $<.001$ & adv \\
$Q$-dial\,$-$\,Marshall, $T{=}2$ (pool.) & 40 & $+.043$ & $[+.019,+.073]$ & $<.001$ & adv \\
fus.\,$-$\,class., Fashion-MNIST & 10 & $+.001$ & $[-.017,+.018]$ & $.86$ & eq \\
fus.\,$-$\,class., synth.\ $\alpha{=}0.75$ & 10 & $-.033$ & $[-.056,-.010]$ & $.04$ & def \\
fus.\,$-$\,class., Iris-3\,/\,MNIST-012 & 5 & \multicolumn{2}{l}{$-.013$/$-.008$; CIs ${\subset}[-.031,.009]$} & $\ge.38$ & inc/eq \\
\bottomrule
\end{tabular}
\ifqmi\end{table}\else\end{table*}\fi

\paragraph{Compute.}
Table~\ref{tab:compute} lists the wall-clock cost for the studies whose runtime we recorded.
We ran all experiments on a single CPU node (AMD EPYC 9474F, $16$ cores allocated, $62$\,GB
RAM, Linux; no GPU). The software versions were Python~3.13.5, NumPy~2.5.1, SciPy~1.18.0,
scikit-learn~1.9.0, Qiskit~2.5.0, Qiskit-Aer~0.17.2, and \texttt{qiskit-addon-cutting}~0.10.0.
Every experiment threads an explicit integer seed through data generation, parameter
initialization, and heads; we state the seed lists per experiment in the table above. The
full result set regenerates from per-experiment scripts, and every figure regenerates from
versioned result files via one plotting script.

\ifqmi\begin{table}[ht]\else\begin{table*}[t]\fi
\caption{Approximate compute per study (single $16$-core CPU node; BLAS pinned to one thread per
process).}
\label{tab:compute}
\centering\scriptsize
\setlength{\tabcolsep}{3pt}
\begin{tabular}{@{}R{3.6cm}lR{3.6cm}R{1.9cm}@{}}
\toprule
study & runs & wall clock & resources \\
\midrule
dense-$\alpha$ diagnostic grid & $104$ & ${\sim}1$\,h & $6$ processes \\
standard benchmarks (4 datasets) & $40$ & ${\sim}50$\,min & $6$ processes \\
classical split-feature fusion & $144$ & ${\sim}3$\,min & $8$ processes \\
trained deep coupling ($L{\le}4$) & $320$ & $9$\,min & $4$ processes \\
second circuit family (ansatz~B: sweep $+$ benchmarks) & $80$ & ${\sim}1$\,h $+$ ${\sim}25$\,min & $6$ processes \\
physical dial (qiskit-addon-cutting) & $25$ & ${\sim}2.5$\,h ($22$\,min per finite-shot run) & $4$ processes \\
uncut reference, $n{=}14$ ($k\in\{1,4,7\}$, $3$ seeds) & $9{+}9$ & $76$--$95$\,min per run at $35$ iterations; $215$--$244$\,min at $150$ & $8$ processes \\
uncut reference, $n{=}16$ & $1$ & ${>}35$\,CPU-h, not converged & $1$ process \\
uncoupled fusion, $n{=}16$--$20$ & $10$ & minutes per run & $1$ process \\
live hardware (both protocols) & $19$ instances & ${\sim}945$ circuits, minutes of QPU time & \texttt{ibm\_fez} \\
\bottomrule
\end{tabular}
\ifqmi\end{table}\else\end{table*}\fi

\section{Use of generative AI}

 A large language model (an AI coding and
writing assistant) was used throughout the development of this work, in three roles:
\textbf{(i) methodology assessment}---discussing and stress-testing experimental designs,
baselines, and controls, and proposing analyses, ablations, and robustness checks that the
authors then specified and ran; \textbf{(ii) implementation}---writing and refactoring the
simulator, experiment, and plotting code released at
\url{https://github.com/p299sing/late-fusion-qml}; and
\textbf{(iii) writing}---drafting and editing manuscript text.

The authors formulated the research questions and hypotheses, made every design and
interpretation decision, ran all experiments, checked the proofs and the reference list, and
verified every number reported in the paper against the released result files. No text,
figure, or numerical result was accepted without author verification. No AI system is listed as an author or cited as a source, and the authors take full
responsibility for the entire content of the paper, including text, figures, references, and
appendices.

\end{document}